%% file: main-paper.tex
\documentclass[sigconf]{acmart}[10pt]

\renewcommand\footnotetextcopyrightpermission[1]{}
\setcopyright{none}

\usepackage[english]{babel}
\usepackage{blindtext}
\usepackage{subcaption}

\usepackage{tikz,pgfplots}
\usetikzlibrary{shapes.arrows}
\pgfplotsset{every tick label/.append style={font=\scriptsize}}

\usepackage{bm}

\usepackage{amsmath,amssymb,amsfonts}

\usepackage{amsthm} 
\usepackage{algorithmic}  
\usepackage{graphicx}
\usepackage{textcomp}  
\usepackage{xcolor}

\usepackage{bbm}
\usepackage{algorithm}
\usepackage{placeins}
\usepackage{physics}

\usepackage[nopar]{lipsum}
\allowdisplaybreaks
\newcommand\numberthis{\addtocounter{equation}{1}\tag{\theequation}}
\newcommand{\myindent}{\noindent\ \hskip\dimen\@ne}

\usepackage{multicol}

\theoremstyle{plain}
\newtheorem{thm}{Theorem}%[section]
\newtheorem{lem}{Lemma}

\newtheorem{cor}{Corollary}

\newtheorem{assumption}{Assumption}

\theoremstyle{remark}

\newtheorem{remark}{Remark}

\newcommand{\E}{\mathbb{E}}

\newcommand{\R}{\mathbb{R}}
\newcommand{\N}{\mathbb{N}}
\newcommand{\cS}{\mathcal{S}} 
\newcommand{\cA}{\mathcal{A}}

\newcommand{\cP}{\mathcal{P}}

\renewcommand{\epsilon}{\varepsilon}
\renewcommand{\phi}{\varphi}

\def\BibTeX{{\rm B\kern-.05em{\sc i\kern-.025em b}\kern-.08em
    T\kern-.1667em\lower.7ex\hbox{E}\kern-.125emX}}

\def\BibTeX{{\rm B\kern-.05em{\sc i\kern-.025em b}\kern-.08em
    T\kern-.1667em\lower.7ex\hbox{E}\kern-.125emX}}    
\theoremstyle{plain}

\theoremstyle{definition}

\theoremstyle{remark}

\DeclareMathOperator*{\argmax}{arg\,max}
\DeclareMathOperator*{\argmin}{arg\,min}

\usepackage{newfloat}
\DeclareFloatingEnvironment[
    fileext=los,
    listname={List of Algorithms},
    name=Algorithm,
    placement=tbhp,
    within=section,
]{myalgorithm}

\begin{document}

\title[Structured Reinforcement Learning for Media Streaming at the Wireless Edge]{Structured Reinforcement Learning for Media Streaming\\ at the Wireless Edge}

% \author{\IEEEauthorblockN{Archana Bura$^*$, Sarat Chandra Bobbili$^*$, Shreyas Rameshkumar$^*$, 
% Dileep Kalathil$^*$, Srinivas Shakkottai$^*$}\\
% \IEEEauthorblockA{
% 	\textit{$^*$Texas A\&M University, College Station, Texas}} \\
% \{archanabura, saratb, shreyasr, dileep.kalathil, sshakkot\}@tamu.edu
% }

\author{Archana Bura$^+$, Sarat Chandra Bobbili$^*$, Shreyas Rameshkumar$^*$, Desik Rengarajan$^*$,
Dileep Kalathil$^*$, Srinivas Shakkottai$^*$}
% \authornote{Note}
\affiliation{%
  \institution{$^+$ University of California at San Diego, $^*$Texas A\&M University}
  %\streetaddress{Address}
  %\city{ College Station} 
  %\state{Texas} 
  %\postcode{Zipcode}
  \country{United States of America}
}
\email{{abura@ucsd.edu},{ saratb, dileep.kalathil, sshakkot}@tamu.edu,{shreyasr1097, desik.29}@gmail.com
}

\renewcommand{\shortauthors}{Bura.et al.}

% \author{Archana Bura$^*$, Sarat Chandra Bobbili$^*$, Desik Rengarajan$^*$, Shreyas Rameshkumar$^*$, 
% Dileep Kalathil$^*$, Srinivas Shakkottai$^*$}
% % \authornote{Note}
% %\orcid{1234-5678-9012}
% \affiliation{%
%   \institution{$^*$Texas A\&M University}
%   %\streetaddress{Address}
%   \city{ College Station} 
%   %\postcode{Zipcode}
% \email{{archanabura, saratb, desik, shreyasr, dileep.kalathil, sshakkot}@tamu.edu
% }

% % The default list of authors is too long for headers}
% \renewcommand{\shortauthors}{Bura.et al.}

\input{00-abstract}

\maketitle

\input{01-introduction}

\input{02-SystemModel}

\input{03-ExistenceThresholdPolicy}

\input{04-NPGConvergence}

\input{05-Simulations}

\input{06-RealSystem}

\bibliographystyle{ACM-Reference-Format}
\bibliography{refs}
\clearpage
%\bluetext{
\input{07-Appendix}

%}
\end{document}

%% file: 00-abstract.tex
\begin{abstract}
Media streaming is the dominant application over wireless edge 
(access) networks. The increasing softwarization of such networks has led to efforts at intelligent control, wherein application-specific actions may be dynamically taken to enhance the user experience. The goal of this work is to develop and demonstrate learning-based policies for optimal decision making to determine which clients to dynamically prioritize in a video streaming setting. We formulate the policy design question as a constrained Markov decision problem (CMDP), and observe that by using a Lagrangian relaxation we can decompose it into single-client problems. Further, the optimal policy takes a threshold form in the video buffer length, which enables us to design an efficient constrained reinforcement learning (CRL) algorithm to learn it. Specifically, we show that a natural policy gradient (NPG) based algorithm that is derived using the structure of our problem converges to the globally optimal policy.  We then develop a simulation environment for training, and a real-world intelligent controller attached to a WiFi access point for evaluation. We empirically show that the structured learning approach enables fast learning.  Furthermore, such a structured policy can be easily deployed due to low computational complexity, leading to policy execution taking only about 15$\mu$s.  Using YouTube streaming experiments in a resource constrained scenario, we demonstrate that the CRL approach can increase quality of experience (QOE) by over 30\%.
\end{abstract}

%% file: 01-introduction.tex
\section{Introduction}
\label{section: introduction}

Video streaming has grown in recent years to occupy about 70\% of downstream mobile data traffic, with YouTube being the most popular~\cite{sandvine23}.  At the same time, increasing softwarization is taking place in both cellular and WiFi domains, allowing fine grained dynamic control of wireless resource usage to support mobile applications via intelligent control~\cite{onf-ric,sunny2017generic}.  Indeed, intelligent control has been shown to enable dynamic priorities over the wireless edge (radio access network) for specific applications that require it~\cite{ko2023edgeric}. 

The growing interest in intelligent control of wireless edge resources raises the question as to whether dynamic control policies for prioritized spectrum access can be used to obtain appreciable benefits in terms of quality of experience (QoE) at the user end?   The answer to this question is nuanced, since such a system would be viable only if the policy can first be computed in some straightforward manner, and it is both simple to implement and robust enough that it can be utilized in a variety of scenarios with temporal variations in the number of clients and their channel qualities.   

The goal of this work is to explore the problem of designing policies for dynamic resource allocation at the wireless edge in the specific context of media streaming to mobile devices.  A diagram illustrating this use case appears in Figure~\ref{fig:overview}, wherein several YouTube sessions are simultaneously supported by a wireless access point.  Each YouTube session has an application state in terms of the video packets buffered and the number of stalls experienced thus far, while the quality of the wireless channel is part of the state of the host mobile device.  This state is communicated back to an intelligent controller, which provides a policy for prioritization of selected YouTube sessions.  The impact of such a policy in terms of QoE is measured at the end-user and communicated back to the controller, thus completing the feedback loop.  Can we show appreciable benefits of intelligent control in this scenario?

\begin{figure}[htbp]
%\vspace{-0.1in}
\begin{center}
\includegraphics[width=3in]{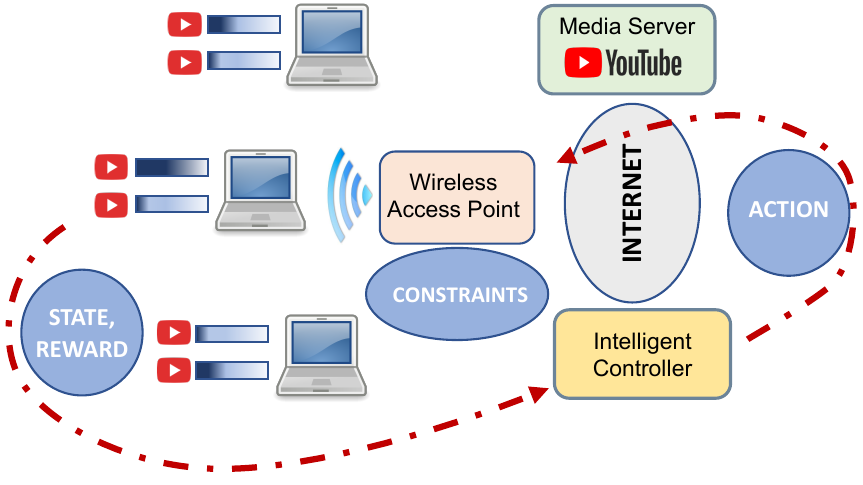}
\caption{Feedback loop in a media streaming application.  The states of the YouTube sessions and channel qualities are communicated to an intelligent controller that determines the service class for each session, accounting for resource constraints.  The impact of this decision on the end-user QoE (reward) is communicated back to the controller.%, which completes the feedback loop.
}
\label{fig:overview} 
\end{center}
%\vspace{-0.1in}
\end{figure}
The feedback loop shown in Figure~\ref{fig:overview} suggests that the problem of deciding how best to allocate resources to clients in order to maximize the QoE across all of them could be posed as a Markov Decision Process (MDP).  The resource allocation problem could take the form of assigning resource blocks to each end-user over sub-frames in a cellular setting, or assigning YouTube sessions to queues of different priority levels in a WiFi setting.  In either case, there are several different service classes, and since resources are limited, there is a constraint across allowable allocations of YouTube sessions to service classes. Thus, we have a constrained MDP (CMDP) as our problem formulation.

Solving such a CMDP is not easy due to the unknown statistics of the many sources of randomness in the system, including interactions between channel quality, the transport protocols used, and the video playout rate.  This suggests that while perhaps the CMDP might give insight into the nature of the optimal policy, the actual policy to be employed needs to be determined via a data-driven approach using constrained reinforcement learning (CRL).  However, RL is a very data hungry approach, and reducing the learning duration is critical if we are to deploy the approach over many different use cases.  Thus, our objective is to first discover the structure of the optimal policy, and use this structure to expedite learning for the CRL process.  Ultimately, our goal is to determine a sufficiently robust policy that will generalize to a variety of scenarios and is deployable in a real-world edge network.

\subsection{Main Results}

We first decompose the centralized CMDP problem across all clients into decentralized single-client problems using Lagrangian relaxation techniques, wherein a Lagrange multiplier acts as a price signal to enforce the constraint.  As in several other queueing problem contexts, we find that the optimal policy of such a single-client problem has a threshold structure.  Specifically, for any given stall count and constraint cost, the decision on whether to assign high priority service to a client only depends on whether or not its video buffer is below a fixed threshold. 

We next consider whether we can provably learn the optimal policy.  Here, the threshold structure of the optimal policy enables us to confine our search to learning the threshold parameter using a single neuron 
(soft threshold), as opposed to a complex function that needs to be approximated with a large neural network.  We develop a primal-dual natural policy gradient algorithm, designed for learning such threshold policies under constraints.  We show that the error in both the objective and constraints decay with time $T$ as $O(1/\sqrt{T}).$   Existing results on the convergence of policy gradient only apply to direct parametrization and soft-max parametrization.  Hence, our methodological contribution is to prove the fast convergence of threshold-parametrized natural policy gradient, which applies to a variety of queueing systems.

% \redtext{We next turn to the important question of learning the optimal policy.  Here, the threshold structure of the optimal policy enables us to confine our search to learning the threshold parameter using a single neuron, as opposed to a complex function that needs to be approximated with a large neural network.  We also observe that the cost-to-go function for a given stall count is unimodal in the threshold parameter, i.e., the cost increases on both sides of the optimal threshold.  This observation allows us to design a three time-scale stochastic approximation algorithm following a policy gradient approach that can converge to the optimal policy.  We update the value function in the first time scale, the policy in the second time scale, and the Lagrange multiplier in the third time scale.}

We next develop a simulation environment to train optimal policies.   We show that the single-client decomposition approach learns at about four times faster than centralized learning while attaining a similar performance. Furthermore, the single neuron structure implies that the inference time for decision making is only about 10 to 15 $\mu$s, as opposed to 50 to 60 $\mu$s taken by an unstructured policy, i.e., the approach is much lower in computational requirements and can be run in a realtime manner.   We also develop a heuristic index-variant that has the advantage of being applicable in spite of a changing number of clients or variability in channel quality.  

Finally, we instantiate a simple intelligent controller platform on which to evaluate the trained policies in a real-world setting.   The controller uses a pub-sub approach to dynamically obtain performance information and to select priority service for particular YouTube clients at a WiFi access point.  
We implement all policies on our intelligent controller and show using over 50 hours of YouTube streaming experiments that we can attain an improvement in QoE of over 30\% as compared to a vanilla policy that does not prioritize any clients, while maintaining a perfect QoE score 60-70\% of the time. We also illustrate the robustness of the index policy to a variety of load and channel conditions while maintaining a higher QoE than other approaches.

\subsection{Related Work}
There is an extensive amount of work on CMDPs focusing on the problems of control~\cite{altman}. In the context of media streaming over wireless channels, several works identify structural properties of the optimal policy~\cite{parandehgheibi2011avoiding,singh2019optimal,hsieh2018heavy} that often take a threshold form.  While \cite{parandehgheibi2011avoiding,singh2019optimal} consider maximizing QoE for a finite system, and \cite{singh2019optimal} utilizes a dual-based decentralization approach, \cite{hsieh2018heavy} focuses on the heavy traffic limit.

Reinforcement learning approaches to find the optimal policies for CMDP problems are  well studied in recent literature. In the context of a tabular setting, recent works focus on characterizing the regret~\cite{liu2021learning,wei2022triple,bura2022dope} or the sample complexity~\cite{hasanzadezonuzy2021learning, hasanzadezonuzy2021model,vaswani2022near} of learning algorithms. There also a number of  works on policy gradient approaches for CMDP \cite{achiam2017constrained, zhang2020first, ding2020npg}. However, these works are for the general CMDPs and they do not exploit the structure of the problem to learn a threshold policy with global convergence guarantees. 

Other work that considers a learning approach to streaming is \cite{roy2021online}, which develops a two time-scale stochastic approximation algorithm \cite{borkar2009stochastic} for an \textit{unconstrained} MDP, and shows certain structural properties of the value function (unimodality), which suggests that a gradient algorithm might converge to the optimal.  However, optimality of the fixed point of gradient descent is not shown.  We take a very different approach via constrained natural policy gradient, and obtain guarantees for fast convergence to the optimal policy.

%The work that is closest to ours is  

%to attain convergence to the global optimum. Our analytical approach towards exploiting problem structure is motivated by this work. Stochastic approximation algorithm with provable guarantees for constrained reinforcement learning was first proposed in  \cite{borkar2005actor} for the tabular setting. This approach has been  extended with function approximation in~\cite{bhatnagar2010actor}, but only with local convergence guarantees. In this work, we exploit the structural properties of the underlying problem with the analytical approach of three time scale stochastic approximation \cite{borkar2005actor} to learn threshold policies with provable global convergence guarantees. 

A variety of systems have been proposed to improve the QoE performance of video streaming. In software defined networks (SDN), assigning the video streaming flows to network links according to path selection algorithms is considered in~\cite{jarschel13sdnyoutube}. VQOA \cite{ramakrishnan15sdnqoe} and QFF \cite{Georgopoulos13openflowqoe} employ SDN to monitor the traffic and adapt the bandwidth assignment of video flows to achieve better streaming performance in the home network. Reinforcement learning has also been applied to a variety of video streaming applications~\cite{mao2017neural,bhattacharyya2019qflow}, which show significant improvements over existing methods.  \cite{mao2017neural} proposes an algorithm Penseive, to train adaptive bit rate (ABR) algorithms to optimize the QoE of the clients.

Finally, \cite{bhattacharyya2019qflow} considers the problem of service prioritization for streaming, and shows empirically that off-the-shelf model-based and model-free RL approaches can result in major improvements in QoE.  While completely heuristic in its approach, the work motivates us to consider the problem of optimal policy design, while exploiting problem structure.   Thus, we provide a simple, structured RL algorithm with provable convergence guarantees, which is empirically able to attain the same performance as more complex deep learning approaches, while only using significantly fewer data samples, and with much less compute needed during runtime.

%% file: 02-SystemModel.tex
\section{PROBLEM FORMULATION}
\label{section: Problem Formulation}

We consider the problem of media streaming from a wireless access point (AP) to a set of clients.  The AP is capable of instantiating multiple service classes by allocating (limited) resources across clients.  We will consider a simple setup with exactly two service classes, referred to as ``high'' and ``low'' service classes.  The high service class provides a better service rate than the low priority class.   An intelligent controller periodically takes decisions on the prioritization of clients to maximize the overall QoE at all clients, under the constraint that only a fixed number may be assigned to the high class.

 We consider a discrete time system with $N$ clients that are connected to the AP and compete for resources. At any time $t$, a media server sends some number of media packets to the AP, which in turns forwards them to the appropriate client, which buffers up this content, while continually attempting to play out a smooth stream.  The media playout is said to \emph{stall} when the media buffer is empty, and is an undesirable event. Stall events cause end-user dissatisfaction, and we will utilize a standard model of this disutility in the video streaming context.
 
 \vspace{0.05in}
 \textbf{State:} We denote the state of client $n$ at time $t$ by $s_n(t) = (x_n(t),y_n(t))$, where $x_n(t)$ is the number of packets in the buffer, and $y_n(t)$ is the number of stalls the client $n$ has encountered until time $t$. We consider a finite state buffer for every client, i.e., $ x_n(t)\in[L]\triangleq \{0,1, \dots, L \}.$   We keep track of stalls up to a finite number, consistent with popular QoE models~\cite{7025402,8013810,8247250} that observe that user perception of stalls does not change significantly after a few stalls. Hence, we choose $y_n(t) \in [M] \triangleq \{0,1,\dots,M\}.$  Thus, the state space of client $n,$ denoted by $\mathcal{S}_n$ is  $[L] \times [M]$ and the joint system state space is given by $\mathcal{S} \triangleq \otimes_{n=1}^N \mathcal{S}_n.$  Client state can be expanded to include its channel quality (which evolves independently of controller actions) without changing our analysis.
 
 \vspace{0.05in}
 \textbf{Actions:} The service class assigned to a client $n$ is denoted by action $a_n(t)$, i.e., $a_n(t) \in \mathcal{A}_n \triangleq \{1,2\}$, where $1$ means the high priority service, and $2$ means low priority service. Depending upon the choice of the service class, client $n$ obtains $A_n(t)$ incoming packets into its media buffer at time $t$.  Let $\mu_n(a)$ represent the service rate obtained by the $n^{th}$ client. We assume that the packet arrival process $A_n(t)$ for every $n$ is a Bernoulli distributed random variable, with a parameter $\mu_n(a)$, depending only upon $a$.
\begin{align*}
\mu_n(a) &= \begin{cases}
\mu_{n,1} & a = 1\\
\mu_{n,2} & a = 2.
\end{cases}
\end{align*}
We assume that $\mu_{n,1} > \mu_{n,2}$, for each $n$. This means that when the client is assigned to high priority service, its service rate is higher than that of the low priority service.  The joint system action space is $\mathcal{A} \triangleq \otimes_{n=1}^N \mathcal{A}_n$. Note that we could easily include video bitrate adaptation into the model by simply increasing the action space to multiple dimensions.  This causes no changes to the approach. We now describe the dynamics of the system. 

\vspace{0.05in}
\textbf{Policy:} The joint policy $\Pi$ is a mapping from joint state space $\mathcal{S}$ to joint action space $\mathcal{A}$, $\Pi:\mathcal{S}\times \mathcal{A} \rightarrow [0,1]$. 

\vspace{0.05in}
\textbf{State Transition Structure:}
The system state has two  component-wise transition processes. The first process corresponds to the media buffer evolution, while the second corresponds to the stall count evolution.  We assume that the media playback process of client $n$, denoted by $B_n(t)$ is distributed according to a Bernoulli distribution with mean $\beta_n < \mu_{n,1},$ i.e., stall-free playout is possible if the client is given a high priority service.   Furthermore, a stall event corresponds to a transition from a non-zero media buffer state to a zero media buffer state (the buffer cannot be negative). 

We assume that the media stream that is currently being played by the client may be terminated at any time with probability $\alpha > 0$ by the end-user (i.e., the user simply stops the playout), which causes a reset of the client state to $(0,0).$  Hence, the possible state transitions from $s_n(t)$ to $s_n(t+1)$ are
\begin{align*}
(x_n(t+1),y_n(t+1)) \hspace{2.1in}\\
= \begin{cases}
(x_n(t),y_n(t)), 
\hspace{5mm}\text{w.p}~ (1-\alpha)P_{n,1}(a_n(t)), \\
(\min\{x_n(t) + 1, L\},y_n(t)),
~ \text{w.p}  ~ (1-\alpha)P_{n,2}(a_n(t)), \\
((x_n(t) - 1)^+,\min\{y_n'(t),M\}), 
\text{w.p} ~ (1-\alpha)P_{n,3}(a_n(t)),\\
(0,0), ~\text{w.p}~~~ \alpha,
\end{cases}
\end{align*}
where $x^+ = \max(x,0)$, $y_n'(t) = y_n(t) + \mathbbm{1}\{x_n(t)=1\}$,
$P_{n,1}(a) = 1 -  P_{n,2}(a) - P_{n,3}(a)$,
$P_{n,2}(a) = \mu_n(a) (1-\beta_n)$, $P_{n,3}(a) = (1-\mu_n(a)) \beta_n$, for $a \in \{1,2\}$, and for every $n$.  The state transition probability of the $n$-th client is denoted by $p_n(s'|s,a)$. The starting state distribution of the $n$-th client is denoted by $\rho_n$. Note that we can include client-specific  transition probabilities based on their individual channels without changing the transition structure.

\vspace{0.05in}
%\bluetext{
\textbf{Notation:}
For the ease of presentation, we summarize the notation used for state transitions throughout the article. Since our state space is discrete and two-dimensional, we assume the following algebraic operations on the state space.
Let $\mathbf{0} \triangleq (0,0)$, $e_x \triangleq (1,0)$ and $e_y \triangleq (0,1)$. For a a given state $s = (x,y) \in \cS$ and $k \in \N_{0}$:
\begin{enumerate}
\item $s + k e_x \triangleq (\min\{x+k,L\},y)$
\item $s - k e_x \triangleq ((x-k)^+,y + \mathbbm{1}_{\{x=k\}}).$
%\item $s + k e_y \triangleq (x,\min\{y+k,M\})$
%\item $s - k e_y \triangleq (x,(y-k)^+).$
\end{enumerate}
Note that the operations above are not associative. For example, $(s - e_x ) + e_x \neq s + (-e_x + e_x)$. This is due to the fact that the number of stalls can increase in a state transition.
%Furthermore, we denote the number of packets in the buffer, and stall count of state $s$ by $s^{(x)}$ and $s^{(y)}$ respectively.
%}

\vspace{0.05in}
\textbf{Reward Structure:}
The client's quality of experience (QoE) depends upon the occurrence of stall events, their duration, and the buffer state. A stall occurs when the media playback stops due to the media buffer becoming empty. Hence, we say that a \emph{stall period} begins at time $t$ if the media buffer becomes empty at $t.$   A simple analytically tractable structure for the instantaneous cost at time $t+1$ takes the form $c(s_n(t),s_n(t+1))$ where the instantaneous media buffer evolution from $x_n(t)$ to $x_n(t+1)$ captures whether a stall is ongoing, and the stall count $y_n(t+1)$ allows us to pick the cost function based on how many stalls have occurred thus far.

Note that we choose to use a cost function, rather than a QoE-based reward function for consistency with a generic MDP/RL approach that focuses on cost minimization.  It is easy to translate the instantaneous cost to an instantaneous reward by simply multiplying by $-1$ and maximizing.  As pointed out above, we can easily include the video bitrate (corresponding to the video resolution) in this cost function but choose not to do so for simplicity.  %In our experiments, we fix the resolution at 1080p so that all clients obtain a high resolution video stream.

We can model a variety of QoE metrics using this approach, such as Delivery Quality Score (DQS) \cite{7025402}, generalized DQS \cite{8013810}, and Time-Varying QoE (TV-QoE) \cite{8247250} or the ITU Standard P.1203 \cite{robitza2017modular}.   The models account for both stall events and video bitrate adaptation and have been validated by laboratory human studies, and we use them to to gain insight into the nature of an appropriate cost function $c(s_n(t),s_n(t+1))$. In all these models, the client starts the video with a QoE of $5$. It decreases when the stall event happens, and further decreases when the stall period continues. The QoE recovers when the stall period ends. The recovery is slow with the increase in the number of stall events the client experiences.

Motivated by these models, we make the following assumptions about the instantaneous cost function $c(\cdot,\cdot)$.
\begin{assumption}\label{assume:CostModel}
$c(\cdot,\cdot)$ satisfies the following conditions: 
\begin{enumerate} 
    \item [\textbf{A1.}] The cost increases with the number of stalls i.e., $c(s, s \pm ke_x ) > c(s',s' \pm ke_x)$ for $s = (x,y)$, $s' = (x,y')$, such that $y > y'$ and $k \in \{0,1\}$ for all $x$. 
    \item [\textbf{A2.}] The cost remains constant during a play period i.e.
    $c(s,s \pm k e_x) = c(s', s' \pm ke_x)$ for $s = (x,y)$, $s' = (x',y)$, such that $k \in \{0,1\}$ for all $x >1,x' >1$.
    %\item [\textbf{A3.}] The cost just before a stall begins is lower than the cost just after the stall begins, which in turn is lower than the cost during the stall period i.e., $c(s,s-e_x) \leq c(s-e_x,s-2e_x) \leq c(s-2e_x,s-2e_x)$ where $s = (2,y)$ for each $y \in [M]$.
    %\item [\textbf{A4.}]  The cost decreases after the stall ends and playout begins again i.e., $ c(s,s) \ge c(s,s+e_x)$ where $s = (0,y)$ for each $y \in [M]$.
    %\item [\textbf{A5.}] The cost difference between successive states for all transitions of same type increases i.e. $c(s+e_x,\eta(s+e_x)) - c(s,\eta(s)) < c(s+2e_x,\eta(s+2e_x)) - c(s+e_x,\eta(s+e_x))$ where $s=(x,y)$, and $\eta(s') \in \{ \mathbf{0} ,s', s'+e_x, s'-e_x\}$.
    \item [\textbf{A3.}] The cost $c(s,0) = c(s',0)$ for all $s,s'$ i.e. the cost of client termination is same for all states. 
\end{enumerate} 
\end{assumption}

% \begin{assumption}\label{assume:CostModel}
% $c(x,x',y)$ satisfies the following conditions: 
% \begin{enumerate}
%     \item [\textbf{A1.}] $c(x,x',y') > c(x,x', y)$ for $y' >y,$ i.e., the cost increases with the number of stalls.
%     \item
%     [\textbf{A2.}]
%     $c(x,x',y) = c(x'',x''',y)$, for every $x,x',x'',x''' >0,$ for a fixed $y,$ i.e., the cost remains constant during a play period.
%     \item [\textbf{A3.}] $c(2,1,y) \leq c(1,0,y+1) \leq c(0,0,y+1),$ i.e., the cost just before a stall begins is lower than the cost just after the stall begins, which in turn is lower than the cost during the stall period.   
%      \item [\textbf{A4.}] $ c(0,0,y+1) \ge c(0,1,y+1)$, for every $y,$ i.e., cost decreases after the stall ends and playout begins again.
%     \item [\textbf{A5.}] $2c(x+1,x',y') \leq c(x+2,x'+1,y) + c(x,(x'-1)^+,y''),$ where $x' \in \{x,x+1\},$ and $y^{\prime} = y +  \mathbbm{1}\{x'=0\}$, $y'' = y+\mathbbm{1}\{x'=1\}$. Moreover, $c(x,x',y) = 0$, if $x' \notin \{x,x+1,(x-1)^+\}$, for every $y$.  This means that the cost difference between successive states increases with the stall duration, which is reasonable, since the QoE of the client drops if the stall period is prolonged.
% \end{enumerate} 
% \end{assumption}

\vspace{0.05in}
\textbf{Constrained Markov Decision Process:} Our resource allocation problem takes the form of the constrained infinite horizon discounted optimization

\begin{align*}
&\min_{\pi \in \Pi}    \mathbb{E}^{\pi} \left[\sum_{n=1}^N \sum_{t=0}^{\infty} \gamma^t c(s_n(t),s_n(t+1))\right]\\
&\text{s.t}~~~ \sum_{n=1}^N g(a_n(t)) \leq K, \quad \forall t,
\end{align*}

where $\gamma < 1$ is the discount factor, and $g(a_n(t)) = \mathbbm{1} \{a_n(t) = 1\}$.  Note that the constraint enforces the condition that only $K$ clients may be allowed high priority service at a given time.

\vspace{0.05in}
\textbf{Relaxed CMDP:}  A commonly used approach to solving CMDPs, starting with seminal work by Whittle~\cite{whittle1988restless} is to relax hard constraints and making them hold in expectation to obtain a near-optimal policy.  The hard constraint can then be re-enforced by an indexing approach wherein the actions leading to the highest value are chosen, while subjected to the constraint.  We follow this approach to soften the constraint to obtain the following problem:

\begin{align} \label{eqn:CMDP}
\begin{aligned}
  &\min_{\pi \in \Pi}  \mathbb{E}^{\pi} \left[\sum_{n=1}^N \sum_{t=0}^{\infty} \gamma^t c(s_n(t),s_n(t+1))\right]\\
&\text{s.t}~~~  \mathbb{E}^{\pi} \left[\sum_{n=1}^N \sum_{t=0}^{\infty} \gamma^t g(a_n(t))\right] \leq \bar{K},
\end{aligned}
\end{align}

where, $\bar{K} = K/(1-\gamma)$, $a(t) \sim \Pi(s(t))$, $a(t) = (a_{n}(t))^{N}_{n=1}$, $s(t) = (s_{n}(t))^{N}_{n=1}$.  The constraint is now the discounted total number of times the high priority service is allocated to all the clients.  
Notice that the policy here depends on the joint state of the system as a whole, and a centralized controller is needed to compute the policy, adding to the complexity of the problem. However, under the assumption of strict feasibility of problem ~\eqref{eqn:CMDP} we can find an approximately optimal policy. 

\begin{assumption}[Slater's condition]\label{eqn:slatercond}
    There exists a policy $\bar{\pi} \in \Pi$ such that  $0 < \xi \leq \bar{K} - \mathbb{E}^{\bar{\pi}} [\sum_{n=1}^N \sum_{t=0}^{\infty} \gamma^t g(a_n(t))]$.
\end{assumption}

%\textbf{Lagrangian formulation of \eqref{eqn:CMDP} and it's decentralized solution:}
An immediate consequence of Assumption~\eqref{eqn:slatercond} is that there exists a primal-dual pair $(\pi^*,\lambda^*)$ achieving strong duality.
For the ease of notation, we denote the individual value and discounted constraint functions of client $n$ under policy $\pi$ by,
\begin{align} \label{eqn:MDPsingle}
J_c^{\pi}(\rho_n) &\triangleq \E_{s_n(0) \sim \rho_n} \left[\mathbb{E}^{\pi} \left[ \sum_{t=0}^{\infty} \gamma^t c(s_n(t),s_n(t+1))|s_n(0)\right]\right],\\
J_g^{\pi}(\rho_n) &\triangleq \E_{s_n(0) \sim \rho_n} \left[\mathbb{E}^{\pi} \left[ \sum_{t=0}^{\infty} \gamma^t g(a_n(t))|s_n(0) \right] \right],
\end{align}

Consider the associated max-min formulation of problem \eqref{eqn:CMDP},
\begin{equation}\label{eqn:saddlepoint}
     \max_{\lambda \in \R_+} \min_{\pi \in \Pi} \sum_{n=1}^N J^{\pi}(\rho_n;\lambda)
\end{equation}
where $J^{\pi}(\rho_n;\lambda) \triangleq  J_c^{\pi}(\rho_n) + \lambda (J_g^{\pi}(\rho_n) - \frac{\bar{K}}{N})$,
$\lambda$ is the non-negative dual variable, and the associated dual function is given by $D(\lambda) \triangleq \min_{\pi \in \Pi} \sum_{n=1}^N J^{\pi}(\rho_n;\lambda)$. Let $\lambda^*$ be the optimal dual-variable i.e. $\lambda^* = \argmax_{\lambda \in \R_+} D(\lambda)$, and $\pi^*$ be an optimal solution to problem~\eqref{eqn:CMDP}. $J^{\pi}(\rho_n;\lambda)$ can be interpreted as the Lagrangian value function for a given client $n$ and the equivalence is exact when $N=1$.
\begin{lem}
Given that assumption~\eqref{eqn:slatercond} holds true, then, $D(\lambda^*) = \sum_{n=1}^N J_c^{\pi^*}(\rho_n)$
\end{lem}
However, finding the optimal joint policy even for a fixed $\lambda$ in the dual function can be computationally hard because the complexity increases exponentially with the number of users.
It can be shown that the joint problem can be decomposed in to $N$ smaller problems \cite{singh2019optimal}.

Let $\pi_{n}: \mathcal{S}_n\times \mathcal{A}_n \to [0,1]$ be the policy of client $n$ where the actions $a_n(t)$ at each time step $t$, are taken independently of the state of clients other than $n$. Denoting the optimal policy for the Lagrangian value function of a given client $n$ and $\lambda$ by,

% \begin{align} \label{eqn:MDPsingle}
% &J^{\pi_{n}}(s;\lambda) \nonumber\\
% &\triangleq \mathbb{E}^{\pi_{n}} [ \sum_{t=0}^{\infty} \gamma^t [c(s_n(t),s_n(t+1)) + \lambda g(a_n(t))]|s_n(0) = s],
% \end{align}
% where $a_n(t) \sim \pi_{n}(s_n(t),\cdot)$, $s_n(t+1) \sim p_{n}(\cdot|s_n(t), a_{n}(t))$, and $s_{n}(0) \sim \rho_n$. 
% The optimal client value function and the optimal client policy are defined as
% \begin{align}\label{eqn:MDPsingle2}
%     J^{\pi^*_{n}}(\rho_n;\lambda) &\triangleq \E_{s \sim \rho_n} J^{\pi^*_{n}}(s;\lambda),\\
%     \pi^{*}_{n}(\lambda) &\triangleq \argmin_{\pi_{n}} \E_{s \sim \rho_n} J^{\pi_{n}}(s;\lambda).\nonumber
% \end{align}
\begin{equation}
    \pi^{*}_{n}(\lambda) \triangleq \argmin_{\pi_{n}} J^{\pi_{n}}(\rho_n;\lambda)
\end{equation}

\begin{thm}~\label{eqn:decentpolthm}
 The joint randomized policy $\pi^* = \otimes \pi_n^*(\lambda^*)$ obtained by individually optimizing each client's Lagrangian value function $J^{\pi}(\rho_n;\lambda)$ for the optimal dual variable $\lambda^*$ is indeed optimal for the original centralized CMDP problem \eqref{eqn:CMDP} i.e., $\sum_{n=1}^N J_c^{\pi^*}(\rho_n) = \sum_{n=1}^N J_c^{\pi_n^*(\lambda^*)}(\rho_n)$
\end{thm}
The proof is similar to that of \cite[Theorem $1$]{singh2019optimal}, for completeness, we provide the details in Appendix~\ref{sec:theorem-1}.

Note that Theorem~\eqref{eqn:decentpolthm} states that it is sufficient to look at decentralized class of policies of the form $\pi = \otimes \pi_n$. 
% \redtext{Henceforth, we refer to the optimal policy as $\pi^* = \otimes \pi_n^*(\lambda^*)$, and use the shorthand notation for the cumulative objective functions $\bar{J}_c^{\pi} \triangleq \sum_{n=1}^N J_c^{\pi_n}(\rho_n), \bar{J}_g^{\pi} \triangleq \sum_{n=1}^N J_g^{\pi_n}(\rho_n)$, $\bar{J}_{\lambda}^{\pi} \triangleq \bar{J}_c^{\pi} + \lambda (\bar{J}_g^{\pi} - \bar{K})$.}

% \begin{thm} \label{thm:Decentralized}
% 1). Suppose $\lambda > 0$ is a price and, for each
% client n, $\pi^*_n$
% is an optimal stationary randomized policy for
% the MDP of client $n$ in \eqref{eqn:MDPsingle}, such that the complimentary slackness condition is satisfied, i.e., either $\lambda = 0$, or the constraint is satisfied with equality in~\eqref{eqn:CMDP}. Then, the joint policy, where each of the clients independently randomize their actions, is optimal for the CMDP~\eqref{eqn:CMDP}.  \\
% 2). There do exist such a Lagrange multiplier $\lambda$ and a set of stationary
% randomized policies $\{\pi_n^*: n = 1, 2, \cdots, N\}$ satisfying the above. Denote that pair by $(\lambda^*,\pi_n^*)$.
% \end{thm}
% The proof is similar to that of \cite[Theorem $1$]{singh2019optimal}, for completeness, we provide the details in Appendix~\ref{sec: Lemmas}.

%% file: 03-ExistenceThresholdPolicy.tex
\section{Structure of the optimal policy}

We begin our analysis by showing two important properties of the optimal policy and value function. In particular, we  show that the optimal policy has a threshold structure. In the next section, we will exploit these structural properties to show that our proposed RL algorithm will provably converge to the globally optimal policy.% in finite time.

\subsection{Threshold Structure}

We first show that the optimal policy for each client has a simple threshold structure, following an approach similar to~\cite{singh2019optimal}. We present this result as a theorem below. 

\begin{thm}\label{thm:threshold}
%\bluetext{
The optimal policy $\pi^{*}_{n}$ for any single client $n \in [N]$ has a  threshold structure. More precisely, for a given state $s=(x,y)$, there exists a threshold function $f^{*}(\cdot)$ such that

\begin{align*}
\pi^{*}_{n}(s,a) &= \begin{cases}
\mathbbm{1}\{x \leq f^{*}(y) \},&\quad \textnormal{if}~~ a = 1 , \\
\mathbbm{1}\{x > f^{*}(y) \}, &\quad \textnormal{if}~~ a=2.
\end{cases}
\end{align*}
%}
\end{thm}
We will use the following results to prove the above theorem. Since the transition function of all the clients have the same properties (even though they may not be identical), we drop the the client index $n$ for ease of notation, and denote the corresponding optimal value function for state $s$, and parameter $\lambda$ by $J^*(s;\lambda)$.

% \begin{lem}\label{lem:nondcr1}
% $J(x+1,y;\lambda) - J(x,y';\lambda)$ is non-decreasing in $x$, for any given $\lambda$, and any fixed $y$, where $y^{\prime} = y + \mathbbm{1}\{x=0\}$.
% \end{lem}

%\bluetext{
\begin{lem}\label{lem:nondcr1}
Suppose $s \triangleq (x+1,y)$. Then, $J^*(s;\lambda) - J^*(s-e_x;\lambda)$ is non-decreasing in $x$, for any given $\lambda$, and $y$.
\end{lem}

\begin{lem} \label{lem:ifflemma1}
The optimal action obtained in state $s \triangleq (x,y)$ is $2$, if and only if, 
$(1-\beta) [J^*(s+e_x;\lambda) - J^*(s;\lambda)] + \beta [J^*(s;\lambda) - J^*(s-e_x;\lambda)] \ge r$, where $r = c_0 - \frac{\lambda}{\gamma(1-\alpha) (\mu_1 - \mu_2)} $, and $c_0 = \frac{1}{\gamma}((1-\beta)[c(s,s)-c(s,s+e_x)]+\beta[c(s,s-e_x) - c(s,s)])$.
\end{lem}
%}
The proofs of Lemma \ref{lem:nondcr1} and Lemma \ref{lem:ifflemma1} are given in Appendix \ref{sec:append-lemma1-2}. 
We now present the proof of Theorem \ref{thm:threshold}.

%\bluetext{
\begin{proof}[Proof of Theorem \ref{thm:threshold}]
From Lemma~\ref{lem:ifflemma1}, we have that the optimal action in state $s \triangleq (x,y)$ is $2$ iff $(1-\beta) [J^*(s+e_x;\lambda) - J^*(s;\lambda)] + \beta [J^*(s;\lambda) - J^*(s-e_x;\lambda)] \ge r,$ where $r$ is defined as above. Moreover, Lemma~\ref{lem:nondcr1} establishes the non-decreasing property of value function differences, for each stall $y$. Together, these imply that for a given stall count $y$, there exists a buffer level $f^*(y)$ such that if the above inequality is satisfied for the first time when the buffer level $x = f^*(y)$, then this inequality is satisfied for all $x > f^*(y)$, and not satisfied for the buffer levels $x \leq f^*(y)$. Hence, the optimal policy is to take action $1$ when buffer level is below $f^*(y)$, and to take action $2$ when the buffer level is above $f^*(y)$, implying that the optimal policy is of threshold type with threshold   $f^*(y)$.
\end{proof}
%}

%% file: 04-NPGConvergence.tex
\section{RL for Learning the Optimal Threshold}

%\bluetext{
In this section, we develop a primal-dual algorithm to solve the constrained Markov decision process problem, and provide results for its convergence to the globally optimal policy. Our implementation of the proposed algorithm follows the methodology of actor-critic algorithms with function approximation. To begin with, we slightly modify our hard-threshold policy to a parameterized soft-threshold policy. For the soft-threshold parametrization, the probability of taking action $a$ in state $s$ corresponding to a threshold parameter vector $\theta$, $\pi_{\theta}(s, a)$ is defined as
\begin{align*}
\pi_{\theta}(s, a) &= \begin{cases}
1 - \frac{1}{1+e^{\theta(s)}} &\quad \textnormal{if}~~ a = 1, \\
\frac{1}{1+e^{\theta(s)}} &\quad \textnormal{if}~~ a = 2,
\end{cases}
\end{align*}
where $\theta:[L] \times [M] \rightarrow \R$.

Note that there exists a one-to-one correspondence between soft-threshold policy parametrization and the hard-threshold function. That is, given a threshold parameter $\theta$, there exists a unique $f$ under the linear relation, $\theta(s) = f(y) - x,~s = (x,y) \in [L] \times [M]$.

In the parameterized setting, we denote the respective value function and the discounted constraint functions by
\begin{align}
&J_c^{\pi_{\theta_n}}(\rho_n) \triangleq \E_{s_n(0) \sim \rho_n} \left[\mathbb{E}^{\pi_{\theta_n}} \left[ \sum_{t=0}^{\infty} \gamma^t c(s_n(t),s_n(t+1)| s_n(0) \right] \right]\\
&J_g^{\pi_{\theta_n}}(\rho_n) \triangleq \E_{s_n(0) \sim \rho_n} \left[\mathbb{E}^{\pi_{\theta_n}} \left[ \sum_{t=0}^{\infty} \gamma^t g(a_n(t))) |s_n(0) \right] \right]
\end{align}

Similarly, the Lagrangian value function for a given client $n$ with initial state distribution $\rho_n$, and $\lambda$  is given by
\begin{equation}
J^{\pi_{\theta_n}}(\rho_n;\lambda)
\triangleq  J_c^{\pi_{\theta_n}}(\rho_n) + \lambda \left(J_g^{\pi_{\theta_n}}(\rho_n) - \frac{\bar{K}}{N}\right).
% &= \mathbb{E}^{\pi_{\theta_n}} [ \sum_{t=0}^{\infty} \gamma^t (c(s_n(t),s_n(t+1)) + \lambda g(a_n(t)))]\\
% &= \mathbb{E}^{\pi_{\theta_n}} [ \sum_{t=0}^{\infty} \gamma^t c(s_n(t),s_n(t+1)] + \lambda \mathbb{E}^{\pi_{\theta_n}} [ \sum_{t=0}^{\infty} \gamma^t g(a_n(t)))]\\
\end{equation}

It is instructive to view $J^{\pi_{\theta_n}}(\rho_n;\lambda)$ as the value function with the cost function $\bar{c}(s_n(t),s_n(t+1)) \triangleq c(s_n(t),s_n(t+1)) + \lambda (g(a_n(t))-K/N)$. Therefore, we denote the corresponding advantage functions for $J^{\pi_{\theta_n}}(\rho_n;\lambda), J_c^{\pi_{\theta_n}}(\rho_n)$ and $J_g^{\pi_{\theta_n}}(\rho_n)$ by $A_{\lambda}^{\pi_{\theta_n}}(s,a)$,$A_c^{\pi_{\theta_n}}(s,a)$, and $A_g^{\pi_{\theta_n}}(s,a)$ respectively. It is easy to show that for any state-action pair $(s,a)$, the following relation holds true:
$A_{\lambda}^{\pi_{\theta_n}}(s,a)= A_c^{\pi_{\theta_n}}(s,a) + \lambda A_g^{\pi_{\theta_n}}(s,a)$.

The above parametrization enables us to use policy gradient methods. In particular, our approach to prove convergence is inspired by the framework of the Natural Policy Gradient (NPG) method for Constrained Markov Decision Processes~\cite{ding2020npg}.

% \redtext{\textbf{Remarks:} We would like to point out a few differences in the analysis of policy gradient methods whose convergence rates are summarized in~\cite[Table 1]{agarwal2021theory}. In particular, we restrict our attention to the class of soft-max parameterized policies owing to its similarity to the structure of the soft-threshold parametrization. Although, one can think of the soft-threshold parametrization as a special case, the convergence analysis does not carry over directly from it's soft-max counterpart. Our reasons to choose NPG over the vanilla policy gradient method (PG) is multi-faceted. To begin with the finite convergence rate guarantees for the soft-max parametrization for the PG method requires one to optimize over a regularized optimization objective whose analysis use the special structure of the soft-max which is not true for the soft-threshold. This is due to the fact that the soft-threshold parametrization does not dependent on the action unlike the soft-max. The same is the case for the NPG-method where the Fischer-information matrix induced by the soft-max policy has a nice structure. However, with some additional algebra one can derive similarly simple parameter updates in the soft-threshold case for the NPG method. Moreover, the convergence guarantees of the regularized version of PG depends on a distribution mismatch coefficient with respect to the optimal policy which is unknown in general.}
%}

\subsection{Convergence of NPG for Threshold Parametrization}
Consider the following primal-dual method to solve the minimization problem~\eqref{eqn:CMDP}.
\begin{align}\label{eqn:pseudoNPGalgo}
\theta_{n,{t+1}} &= \theta_{n,t} - \eta_1 F(\theta_{n,t})^{\dagger} \nabla_{\theta_n}J^{\pi_{\theta_{n,t}}}(\rho_n;\lambda_t) \nonumber\\
\lambda_{t+1} &= \cP_{\Lambda} \left(\lambda_t + \eta_2 \sum_{n=1}^N \nabla_{\lambda}J^{\pi_{\theta_{n,t}}}(\rho_n;\lambda_t)\right), 
\end{align}
where $\theta_{n,t}$ is the threshold parameter of client $n$ at time $t$, 
%$F_{\rho}(\theta_{n,t}) \triangleq \E_{(s,a)\sim \pi_{\theta_{n,t}}}[\nabla_{\theta_n} \log\pi_{\theta_{n,t}}(a|s) \nabla_{\theta_n} \log\pi_{\theta_{n,t}}(a|s)^\top]$ 
$F(\theta) \triangleq \E_{(s,a)\sim \pi_{\theta}}[\nabla_{\theta} \log\pi_{\theta}(a|s) (\nabla_{\theta} \log\pi_{\theta}(a|s))^\top]$ is the Fisher information matrix induced by the policy $\pi_{\theta}$, $A^\dagger$ is the Moore-Penrose inverse of matrix $A$, and $\cP_{\Lambda}$ is the projection operator onto the interval $\Lambda$. $\eta_1$ and $\eta_2$ are constant step-sizes.

Note that Algorithm~\eqref{eqn:pseudoNPGalgo} is a straightforward extension of the standard primal-dual approach to solve the saddle point problem~\eqref{eqn:saddlepoint}. The primal variables here are the threshold parameters of each client $n \in [N]$. At each iteration $t+1$, the threshold parameters of client $n$ follow an NPG update independent of the others. Intuitively, the update makes sense as it is known from Theorem~\ref{eqn:decentpolthm} that it is sufficient to consider client policies that are independent of each other, but the respective value functions are coupled only through the dual variable $\lambda_t$ from the previous iteration $t$. The dual variable update at each iteration, however, uses the sum of gradients from each client $n$. It is easy to see that the our approach reduces to the standard primal-dual method when $N=1$.

%%%%%%%%%%%%%%%%%%%%%%%%%%%%%%%%%%%%%%%%%%%%%%%%%%%%%%%%%%%%%%%%%%%%%%%%%%%%%

The following lemma shows that the soft-threshold parametrization has a rather simple update rule under the natural policy gradient method. Denoting actions $1$, $2$ by $a_1$ and $a_2$, respectively.

\begin{lem}\label{lem:npgupdateform}
%Let $\cA$ be the set of actions such that $\abs{\cA} = 2$. Then, for 
For the soft-threshold policy parametrization, the following are equivalent:
\begin{enumerate}
    \item $\theta_{n,{t+1}} = \theta_{n,t} - \eta_1 F(\theta_{n,t})^{\dagger} \nabla_{\theta}J^{\pi_{\theta_{n,t}}}(\rho_n;\lambda_t)$\\
    \item $\theta_{n,{t+1}}(s) = \theta_{n,t}(s) - \frac{\eta_1}{(1 - \gamma)} [A^{\pi_{\theta_{n,t}}}_{\lambda_t}(s,a_1) - A^{\pi_{\theta_{n,t}}}_{\lambda_t}(s,a_2)].$
\end{enumerate}
\end{lem}
\begin{proof}
It is easy to show that for the soft-threshold policy class,
\begin{align}
    &\frac{ \partial{\log \pi_{\theta}(s,a)}}{\partial{\theta_{s'}}} \nonumber\\
    &= \mathbbm{1}\{s' = s\} (\mathbbm{1}\{a = a_1\} - \mathbbm{1}\{a = a_2\}) (1 - \pi_{\theta}(s,a)).
\end{align}
Therefore,
\begin{equation}\label{eqn:prop1}
\pi_{\theta}(s,a_1) \nabla_{\theta} \log \pi_{\theta}(s,a_1) + \pi_{\theta}(s,a_2) \nabla_{\theta} \log \pi_{\theta}(s,a_2) = 0,
\end{equation}
for any $s$, and
\begin{equation}\label{eqn:prop2}
 \nabla_{\theta} \log \pi_{\theta}(s,a_1) - \nabla_{\theta} \log \pi_{\theta}(s,a_2) = e_s,
\end{equation}
where $e_s \in \R^{\abs{\cS}}$ such that the $i$-th entry is $1$ if $i=s$ and $0$ otherwise. 
%\redtext{We now use the fact that the direction of primal update in~\eqref{eqn:pseudoNPGalgo} parallels the solution of}

We now consider the following quadratic minimization problem:
\begin{align}
    \min_{w \in \R^{|\cS|}} \E_{s \sim d_{\rho_n}^{\pi_{\theta_{n,t}}}, a \sim \pi_{\theta_{n,t}}(s,\cdot)}
    &[(A^{\pi_{\theta_{n,t}}}_{\lambda_t}(s,a) \nonumber\\
    &- w\cdot \nabla_{\theta}  \log \pi_{\theta_t}(s,a))^2].
\end{align}
The optimal solution to the above minimization problem is given by $w^* = (1-\gamma)F(\theta_{n,t})^{\dagger} \nabla_{\theta}J^{\pi_{\theta_{n,t}}}(\rho_n;\lambda_t)$. 

Notice that the direction of primal update in equations~\eqref{eqn:pseudoNPGalgo} parallels this solution. 
Using this fact, we have,
\begin{align*}
&\E_{s \sim d_{\rho_n}^{\pi_{\theta_{n,t}}}, a \sim \pi_{\theta_{n,t}}(s,\cdot)}[(A^{\pi_{\theta_{n,t}}}_{\lambda_t}(s,a) - w\cdot \nabla_{\theta}  \log \pi_{\theta_{n,t}}(s,a))^2]\\
&= \E_{s \sim d_{\rho_n}^{\pi_{\theta_{n,t}}}}[\pi_{\theta_{n,t}}(s,a_1) (A^{\pi_{\theta_{n,t}}}_{\lambda_t}(s,a_1))^2 \nonumber\\
&+ \pi_{\theta_{n,t}}(s,a_2) (A^{\pi_{\theta_{n,t}}}_{\lambda_t}(s,a_2))^2 \nonumber\\
&- 2 ((A^{\pi_{\theta_{n,t}}}_{\lambda_t}(s,a_1) - A^{\pi_{\theta_{n,t}}}_{\lambda_t}(s,a_2))\nonumber\\
&* \pi_{\theta_{n,t}}(s,a_1) w\cdot \nabla_{\theta} \log \pi_{\theta_{n,t}}(s,a_1)) \nonumber\\
&+ \pi_{\theta_{n,t}}(s,a_1) w\cdot \nabla_{\theta}\log \pi_{\theta_{n,t}}(s,a_1)(w\cdot e_s)]\\
&= \E_{s \sim d_{\rho_n}^{\pi_{\theta_{n,t}}}}[\pi_{\theta_{n,t}}(s,a_1) (A^{\pi_{\theta_{n,t}}}_{\lambda_t}(s,a_1))^2 \nonumber\\
&+ \pi_{\theta_{n,t}}(s,a_2) (A^{\pi_{\theta_{n,t}}}_{\lambda_t}(s,a_2))^2] \nonumber\\
&+ \E_{s \sim d_{\rho_n}^{\pi_{\theta_{n,t}}}}[\pi_{\theta_{n,t}}(s,a_1) \pi_{\theta_{n,t}}(s,a_2) w\cdot e_s \nonumber\\
&(w\cdot e_s- 2 (A^{\pi_{\theta_{n,t}}}_{\lambda_t}(s,a_1) - A^{\pi_{\theta_{n,t}}}_{\lambda_t}(s,a_2)))]
\end{align*}
The first equality follows from expanding the quadratic expansion and~\eqref{eqn:prop1}, and the last equality from~\eqref{eqn:prop2}. It is easy to see that the optimal solution $w^*$ is also the optimal minimizer for the second term in the last equality. Using the fact that the stochastic policy $\pi_{\theta_{n,t}}(s,a) > 0$ for all state-action pairs, we have $w^*\cdot e_s = (A^{\pi_{\theta_{n,t}}}_{\lambda_t}(s,a_1) -A^{\pi_{\theta_{n,t}}}_{\lambda_t}(s,a_2)).$
\end{proof}

An immediate consequence of Lemma~\ref{lem:npgupdateform} is the following form for the policy updates.

\begin{cor}\label{cor:politerate}
The policy iterates for algorithm~\eqref{eqn:pseudoNPGalgo} are as follows:
\begin{align}
\pi_{\theta_{n,t+1}}(s,a) &= \pi_{\theta_{n,t}}(s,a) \frac{\exp(\frac{-\eta_1}{1-\gamma}A^{\pi_{\theta_{n,t}}}_{\lambda_t}(s,a))}{Z_t(s)},\\
\text{where}~~ Z_t(s) &\triangleq \sum_{a \in \cA} \pi_{\theta_{n,t}}(s,a) \exp(\frac{-\eta_1}{1-\gamma}A^{\pi_{\theta_{n,t}}}_{\lambda_t}(s,a)).
\end{align}
\end{cor}

The proof of Corollary~\ref{cor:politerate} is given in Appendix ~\ref{sec:append-lemma1-2}.\\

% \begin{proof}
%     From lemma ~\ref{lem:npgupdateform}, we have  $\theta_{n,t+1}(s) = \theta_{n,t}(s) - \frac{\eta_1}{(1 - \gamma)} [A^{\pi_{\theta_{n,t}}}_{\lambda_t}(s,a_1) - A^{\pi_{\theta_{n,t}}}_{\lambda_t}(s,a_2)]$. By definition, for a given state $s$ the probability of taking action $a_1$,
%     \begin{align*}
%         &\pi_{\theta_{n,t+1}}(s,a_1) = \frac{e^{\theta_{n,t+1}(s)}}{1 + e^{\theta_{n,t+1}(s)}}.\\
%         &= \frac{e^{\theta_{n,t}(s)} e^{\frac{-\eta_1 (A^{\pi_{\theta_{n,t}}}_{\lambda_t}(s,a_1)-A^{\pi_{\theta_{n,t}}}_{\lambda_t}(s,a_2))}{1-\gamma}}}{1 + e^{\theta_{n,t}(s)} e^{\frac{-\eta_1 (A^{\pi_{\theta_{n,t}}}_{\lambda_t}(s,a_1)-A^{\pi_{\theta_{n,t}}}_{\lambda_t}(s,a_2))}{1-\gamma}}}.\\
%         &= \frac{e^{\theta_{n,t}(s)} e^{\frac{-\eta_1 A^{\pi_{\theta_{n,t}}}_{\lambda_t}(s,a_1)}{1-\gamma}}}{e^{\frac{-\eta_1 A^{\pi_{\theta_{n,t}}}_{\lambda_t}(s,a_2)}{1-\gamma}} + e^{\theta_{n,t}(s)} e^{\frac{-\eta_1 A^{\pi_{\theta_{n,t}}}_{\lambda_t}(s,a_1)}{1-\gamma}}}\\
%         &= \frac{ \pi_{\theta_{n,t}}(s,a_1) e^{\frac{-\eta_1 A^{\pi_{\theta_{n,t}}}_{\lambda_t}(s,a_1)}{1-\gamma}}}{ \pi_{\theta_{n,t}}(s,a_2) e^{\frac{-\eta_1 A^{\pi_{\theta_{n,t}}}_{\lambda_t}(s,a_2)}{1-\gamma}} + \pi_{\theta_{n,t}}(s,a_1) e^{\frac{-\eta_1 A^{\pi_{\theta_{n,t}}}_{\lambda_t}(s,a_1)}{1-\gamma}}}\\
%         &= \pi_{\theta_{n,t}}(s,a_1) \frac{\exp(\frac{-\eta_1}{1-\gamma}A^{\pi_{\theta_{n,t}}}_{\lambda_t}(s,a_1))}{Z_t(s)}.\\
%     \end{align*}
% The fourth equality follows from normalization with $1 + e^{\theta_{n,t}(s)}$. Similarly, the policy update for action $a_2$.
% \end{proof}
Evaluating the gradient with respect to the dual variable in the dual update, the equivalent primal-dual method to algorithm~\eqref{eqn:pseudoNPGalgo} is given by
\begin{align}\label{eqn:pseudoNPGalgo2}
&\theta_{n,{t+1}}(s) = \theta_{n,t}(s) -  \frac{\eta_1}{(1 - \gamma)} [A^{\pi_{\theta_{n,t}}}_{\lambda_t}(s,a_1) - A^{\pi_{\theta_{n,t}}}_{\lambda_t}(s,a_2)] \nonumber\\
&\lambda_{t+1} = \cP_{\Lambda}\left(\lambda_t + \eta_2 \left(\sum_{n=1}^N J_g^{\pi_{\theta_{n,t}}}(\rho_n) - \bar{K}\right)\right). 
\end{align}
%%%%%%%%%%%%%%%%%%%%%%%%%%%%%%%%%%%%%%%%%%%%%%%%%%%%%%%%%%%%%%%%%%%%%%%%%%%%%%%%%%%%%%%%%%%%%%%%%%%%%%%%%%%%%%%%%%%%%%%%%%%%%%
We have the following convergence guarantee for algorithm \eqref{eqn:pseudoNPGalgo2}.
\begin{thm}\label{thm:finiteconvthm}
Let $\Lambda = [0,\frac{2N}{(1-\gamma) \xi}]$, $\rho_n \in \triangle_{\cS_n}, \theta_{n,0} = \mathbf{0}$, be the starting state distribution and threshold parameter initialization of client $n$, and $\lambda_0 = 0$, be the initial dual variable. For the particular choice of $\eta_1 = \log|\cA|$, $\eta_2 = \frac{1-\gamma}{N\sqrt{T}}$, the iterates generated by the algorithm \eqref{eqn:pseudoNPGalgo} satisfy
\begin{align}
    \frac{1}{T} \sum_{t = 0}^{T-1} \sum_{n = 1}^N (J_c^{\pi_{\theta_{n,t}}}(\rho_n) - J_c^{\pi_n^*}(\rho_n)) &\leq \frac{4N}{(1-\gamma)^2 \sqrt{T}}. \\
        \left(\frac{1}{T} \sum_{t = 0}^{T-1} \sum_{n = 1}^N J_g^{\pi_{\theta_{n,t}}}(\rho) - \bar{K} \right)^+ &\leq  \frac{(2/\xi + 4\xi)N^2}{(1-\gamma)^2\sqrt{T}}.
\end{align}
\end{thm}
The proof sketch is as follows.  
%technique we follow is similar to that of ~\cite{ding2020npg}. 
We first show that the NPG method has a simple parameter update for the threshold policy class. We then prove that the average cost function generated by the iterates of the algorithm converges to the global optimal value, and bound the constraint violation. The details are provided in appendix.  

Note that the quadratic scaling of constraint error with the number of users $N$ arises due the resources being fixed. It is easy to show that if the ability to support high-quality service, $\bar{K}$ scales linearly with $N$, the constraint error will also only be linear in $N.$

%Appendix~\ref{sec:convlogregobj}.

\textbf{Remark:} The notion of soft-threshold parametrization coupled with natural policy gradient is motivated by constrained RL using soft-max parametrization with natural policy gradient~\cite{ding2020npg}.    However, the convergence analysis is not directly transferable from one to the other, since the structure of Fischer-information matrix induced by the soft-max policy is different from that of the soft-threshold policy, which translates to different primal dual updates in equation~\eqref{eqn:pseudoNPGalgo2}.  Nevertheless, the overall flow of the proof is adaptable to the soft-threshold case, and we obtain a similar rate of convergence.

\subsubsection*{Algorithm Design}
Algorithm \ref{algo:SA3} summarizes our approach~\eqref{eqn:pseudoNPGalgo} consistent with the standard actor-critic implementations of RL algorithms. We iteratively update the value function (and hence the advantage function), the policy parameter, and the Lagrange multiplier.  At each time step, we run one actor-critic update for each client. In particular, for each client $n \in N$, we sample an action according to the current policy of the client to update the value function of the current state $J_{n,t+1}(s_n(t))$ from which the advantage function for each action $a \in \cA$, $A_{n,t}(s_n(t),a)$ is computed. The threshold parameter for the current state is updated according to~\eqref{eqn:pseudoNPGalgo}. The value functions and parameters of all but the current state remain unchanged. After each client updates its value function and its policy, the experienced actions are retrieved, and the Lagrange multiplier is updated using the constraint formulation.

\subsubsection*{Note} In our implementation we update the threshold parameter with just the advantage of action $a_1$ instead of the difference between advantages. This variant of the threshold update is harmless and does not change the direction of the gradient since the advantages of different actions have opposite signs owing to the fact that expected advantage is always zero.

\begin{algorithm}
\caption{Threshold Natural Policy Gradient Algorithm}
\label{algo:SA3}
\begin{algorithmic}[1]
\STATE \textbf{Initialize} %Fix an arbitrarily large $T$.
Set threshold $\theta_{n,0} = 0$, starting state $s_n(0) \sim \rho_n$ for every client $n \in N$, $\eta_1 = \ln{2}$, $\eta_2 = \frac{1-\gamma}{N\sqrt{T}}$, and $\lambda_0 = 0$. Set state visitation count $\
\eta(n,s) = 1~~\forall s \in \cS_n$ and for each client $n \in N$.
 \FOR{$t =0, 1, \dots,$}
 \FOR{Each client $n = 0,1,\dots,N$}
  \STATE Update the count $\eta(n,s_n(t)) \gets \eta(n,s_n(t)) + 1.$
  \STATE Take the action $a_n(t) \sim \pi_{\theta_{n,t}}(s_n(t),.)$.
  \STATE  Obtain the next state $s_n(t+1)$, and cost $c(s_n(t),s_n(t+1))$, $g(a_n(t))$  from the environment.
  \STATE  Update the value function according to
 \begin{equation*} 
    \begin{aligned}
     &J_{n,t+1}(s_n(t)) = J_{n,t}(s_n(t)) + \frac{1}{\eta(n,s_n(t))} \left[  \lambda_t g(a_n(t))  \right.\\
      +&\left. [c(s_n(t),s_n(t+1)) +\gamma J_{n,t}(s_n(t+1))] - \right.\\
      &\left. \hspace{40mm} J_{n,t}(s_n(t)) \right], \\
     &J_{n,t+1}(s') = J_{n,t}(s'), \forall s' \neq s_n(t).
\end{aligned}
 \end{equation*}
 %\STATE Update the advantage function $A_{n,t}(s_n(t),a)$ for each action $a \in \{a_1,a_2\}$ 
 \STATE  Update the threshold parameter
\begin{align*}
 &\theta_{n,{t+1}}(s_n(t)) \gets \theta_{n,t}(s_n(t))\\
 &-  \frac{\eta_1}{(1 - \gamma)} [A_{n,t}(s_n(t),a_1) - A_{n,t}(s_n(t),a_2)],\\
 &\theta_{n,{t+1}}(s) = \theta_{n,t}(s), \forall s \neq s_n(t).
\end{align*}
%$\forall ~s \in \cS_n$
\ENDFOR
\STATE Update the Lagrange multiplier\\
$  \lambda_{t+1} = \cP_{\Lambda}(\lambda_t + \eta_2 [\sum_{n=1}^N g(a_n(t))  - \bar{K}]).
$
\ENDFOR
\end{algorithmic}
\end{algorithm}

%% file: 05-Simulations.tex
\begin{figure*}[htbp]
\centering
\begin{minipage}{.32\textwidth}
\centering
%\captionsetup{width=\linewidth}
\includegraphics[width=1\columnwidth]{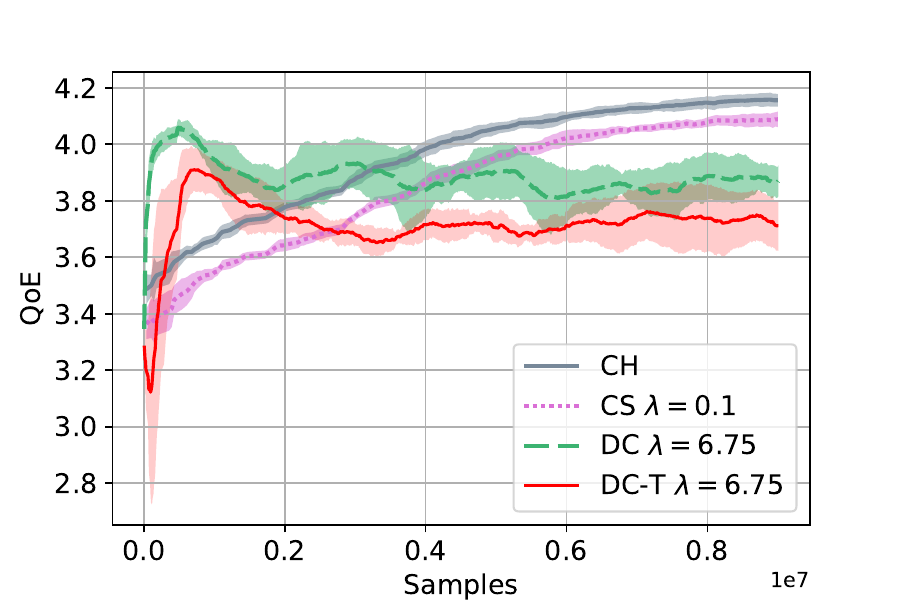}
\caption{Training in simulation}
\label{fig:sim_train_qoe}
\end{minipage}
\hspace{-0.005in}
\begin{minipage}{.32\textwidth}
\centering
\includegraphics[width=1\columnwidth]{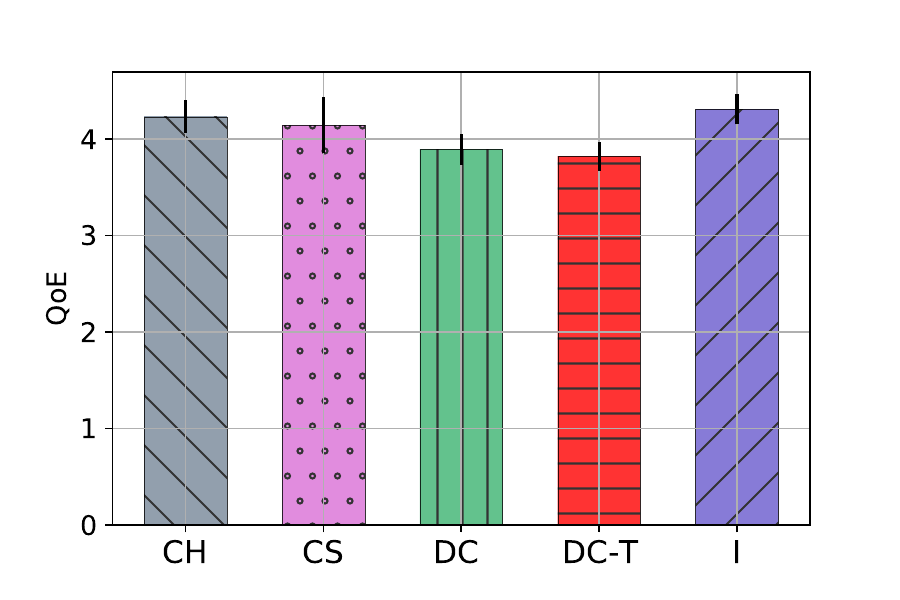}
\caption{QoE in simulation}
\label{fig:sim_qoe}
\end{minipage}
\hspace{-0.05in}
\begin{minipage}{.32\textwidth}
\centering
\includegraphics[width=1\columnwidth]{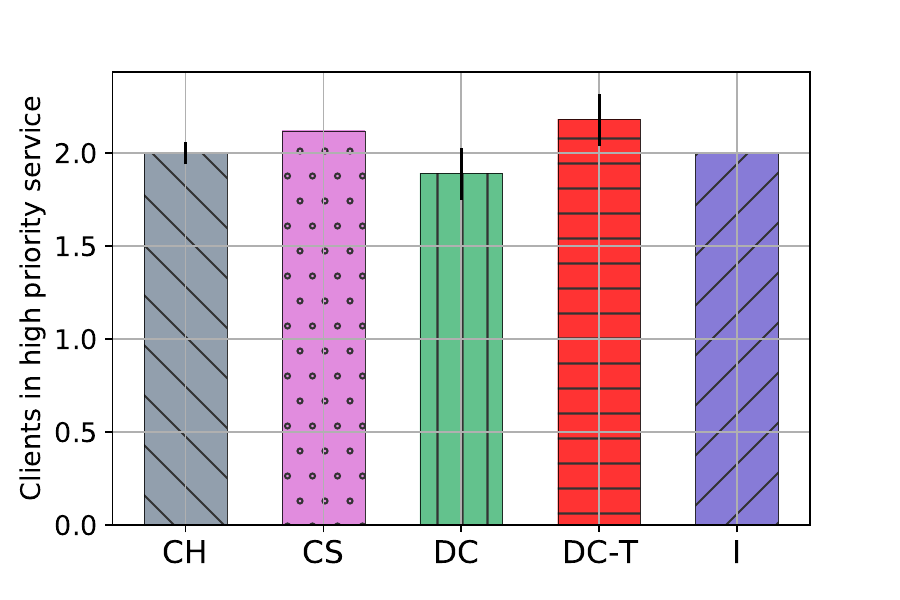}
\vspace{-0.15in}
\caption{Clients in high priority service}
\label{fig:sim_HPS}
\end{minipage}
\vspace{-0.1in}
\end{figure*}

\begin{figure}
\includegraphics[width=0.9\columnwidth]{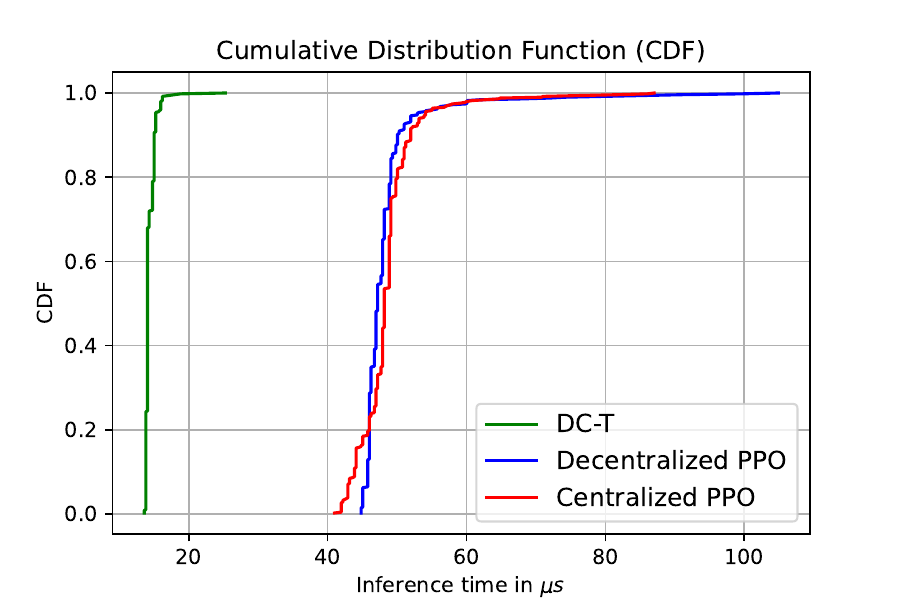}
\caption{Inference time for DC-T and PPO algorithms for $6$ clients}
\label{fig:compute_time}
\vspace{-0.1in}
\end{figure}

\section{Simulation-based Training}
\label{section:simulation}
We first train our system by developing a simulator that captures the dynamics of the system shown in Figure~\ref{fig:overview}. The simulator consists of a wireless access point (AP) with two service classes, each has a fixed bandwidth (e.g., 12 Mbps for ``high'', and 4 Mbps for ``low'') and $N$ clients that desire service, with our typical setting being 6 or fewer clients.  An intelligent controller assigns clients to one of the service classes, and the available bandwidth of that class is partitioned equally among the clients assigned to it, with the realized bandwidth being drawn from a normal distribution around the mean value to model environmental randomness.  For instance, if we set a constraint of at most 2 occupants of the ``high'' service class, each would get a mean throughput of 6 Mbps per-client in the example.  Each client has a state $(x,y),$ consisting of the video buffer length and the number of stalls.  The QoE of the client is calculated using the DQS model described in Section~\ref{section: Problem Formulation}.   Note that we use ``reward'' using the QoE, rather than ``cost'' as used in the analytical model, since most RL implementations use this approach (i.e., we multiply ``cost'' by -1 and maximize instead of minimizing).  

The state of the system as a whole is the union of the states of all $6$ clients, which in the simulator is uncountably infinite, since the video buffer is a real number and the number of stalls can be unbounded.  The action is the service assignment of each client. Hence, there are a total of $2^6$ possible actions ($2$ for each client).  As indicated in Section~\ref{section: Problem Formulation}, when the number of clients in each service class is fixed, the system can be treated as $N$ independent single-client systems, each having two service classes with appropriately scaled-down bandwidth.  Intuitively, training on this system should be faster, since we obtain $N$ [state, action, next-state, reward] samples per time step here, as opposed to just one in the joint system case.  Furthermore, solving using the Dual approach introduces a Lagrange multiplier $\lambda,$ which also has to be learned by performing a hyperparameter search to ensure that when the controller follows the optimal policy the number of clients in each service class satisfies the service class constraints.  

\subsection{Algorithms}
\textbf{Vanilla (V):}  This is a simple policy that merges all available throughput into a single service class and shares it equally.  It is equivalent to CSMA-based random access in WiFi, or a simple round-robin scheduler in cellular.  We will use this as a base policy in the real-world system.%, but do not simulate it.

\textbf{Greedy Buffer (GB):} This algorithm awards high priority service to the clients with the lowest video buffer state, subject to resource constraints.   The algorithm follows the general theme of max-weight~\cite{TasEph92} and min-deficit~\cite{HouBor09} being throughput or timely-throughput optimal in queueing systems.  However, it does not account for the dependence of QoE on stall count.  We use it as a well-established algorithm with good performance in the real-world evaluation.  

\textbf{Centralized Hard PPO (CH):} This algorithm represents recent efforts at applying off-the-shelf RL techniques in media streaming applications, such as those presented in \cite{mao2017neural,bhattacharyya2019qflow}.  In particular, QFlow~\cite{bhattacharyya2019qflow} shows that unstructured Deep Q-Learning (DQN) performs really well in a context much like ours but requires a long training duration.  Policy gradient algorithms are the current norm, having much better empirical performance than DQN, and so we pick Proximal Policy Optimization (PPO)~\cite{schulman2017proximal}.  We implement it with a hard constraint of 2 clients with ``high'' service and 4 with ``low'' service.

\textbf{Centralized Soft PPO (CS):}  This algorithm implements the Lagrangian relaxation by adding a penalty $\lambda$ for accessing ``high'' service.  We use the same PPO algorithm, but do not impose a hard constraint as in CH.  Rather, the algorithm is trained using a hyperparameter search over $\lambda$ such that we have an average of 2 clients in the ``high'' service class.

\textbf{Decentralized PPO (DC):} This algorithm takes advantage of the conditional independence property described in Section~\ref{section: Problem Formulation} that enables division into $N$ independent systems, along with a Lagrangian relaxation of the constraints.  We now use the PPO algorithm on an individual client basis, but since we obtain 6 samples per step, it can train much faster.  However, we do not impose any structure on the policy class.  Again, we need to find the appropriate penalty $\lambda$ to enforce the constraint.

\textbf{Decentralized Threshold (DC-T)} This algorithm is similar to DC in utilizing division into $N$ systems for faster learning, but also imposes a threshold structure of the optimal algorithm.  It requires only one neuron (logistic function) to represent a threshold policy, which makes for simple implementation. Like DC, it too can train faster than the centralized approaches, and also needs an appropriate penalty term to enforce constraints.  

\textbf{Index (I):} This algorithm follows the philosophy behind the Whittle Index~\cite{whittle1988restless} in converting a soft-constrained into a hard-constrained and robust policy.  We order the states of clients based on their values obtained from DC-T, and provide ``high'' service the two with the highest value. We will see that this can be used regardless of the number of clients, or the quality of their channels.

\subsection{Training and Evaluation}
Figure \ref{fig:sim_train_qoe} shows the evaluation of the algorithms on a joint system during training (averaged over $5$ random seeds).  We observe that decentralized algorithms converge over four times faster than centralized algorithms as they exploit the structure of the environment. The performance difference between the trained centralized and decentralized algorithms is negligible as seen in Figure~\ref{fig:sim_qoe} (averaged over $100$ runs).  The performance of DC-T is on par with the best performing algorithms, which confirms our hypothesis that the optimal policy is indeed a threshold policy.  The index policy, I that is hard-constrained version of DC-T shows similarly high performance.  Figure~\ref{fig:sim_HPS} (averaged over $100$ runs) shows the number of clients in the high priority queue during evaluation of the soft constrained algorithms is near 2, and so confirms our choice of $\lambda$.
%\textcolor{blue}{
Finally, Figure~\ref{fig:compute_time} shows the efficacy of the DC-T algorithm in deployment in terms of inference times. We see that the PPO algorithms take roughly 50 to 60 $\mu$s, while the DC-T algorithm takes only about 10 to 15 $\mu$s for inferring the decisions for the $6$ client system.
%}

%% file: 06-RealSystem.tex
\section{Real-World Evaluation}
\begin{figure*}[htbp]
\begin{minipage}{.32\textwidth}
\centering
%\captionsetup{width=\linewidth}
\includegraphics[width=1\columnwidth]{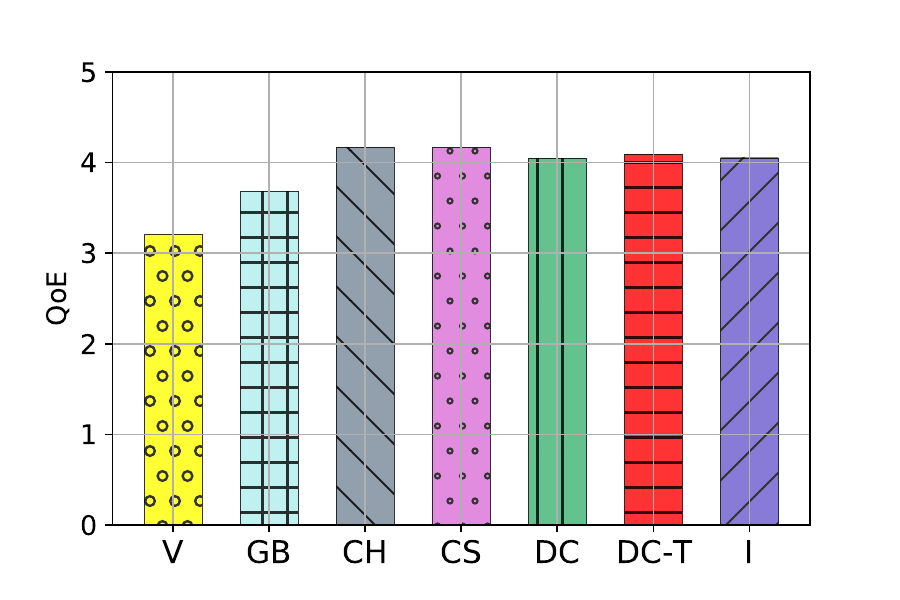}
\caption{Comparison of average QoE}
\label{fig:scenario1_1}
\end{minipage}\hfill
\begin{minipage}{.32\textwidth}
\centering
\includegraphics[width=1\columnwidth]{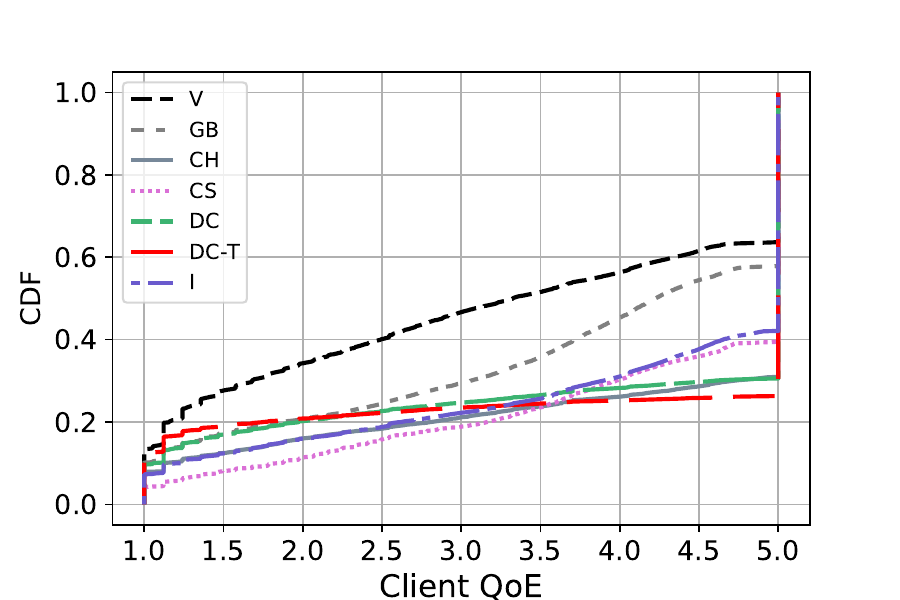}
\caption{Comparison of QoE CDF }
\label{fig:scenario1_2_qoeCDF}
\end{minipage}\hfill
\begin{minipage}{.32\textwidth}
\centering
\includegraphics[width=1\columnwidth]{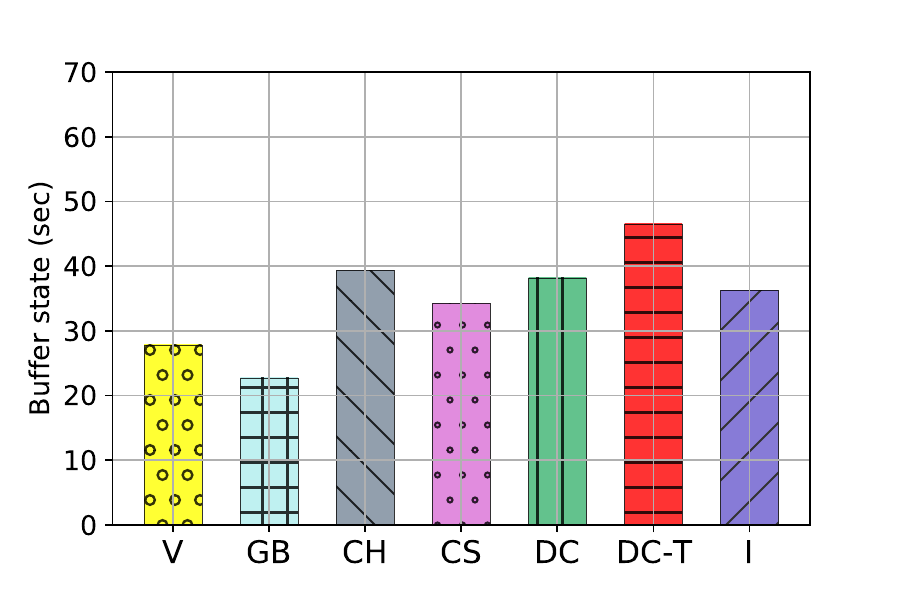}
\caption{Comparison of average buffer}
\label{fig:scenario1_3}
\end{minipage}\hfill
\vspace{-0.1in}

\begin{minipage}{.32\textwidth}
\centering
\includegraphics[width=1\columnwidth]{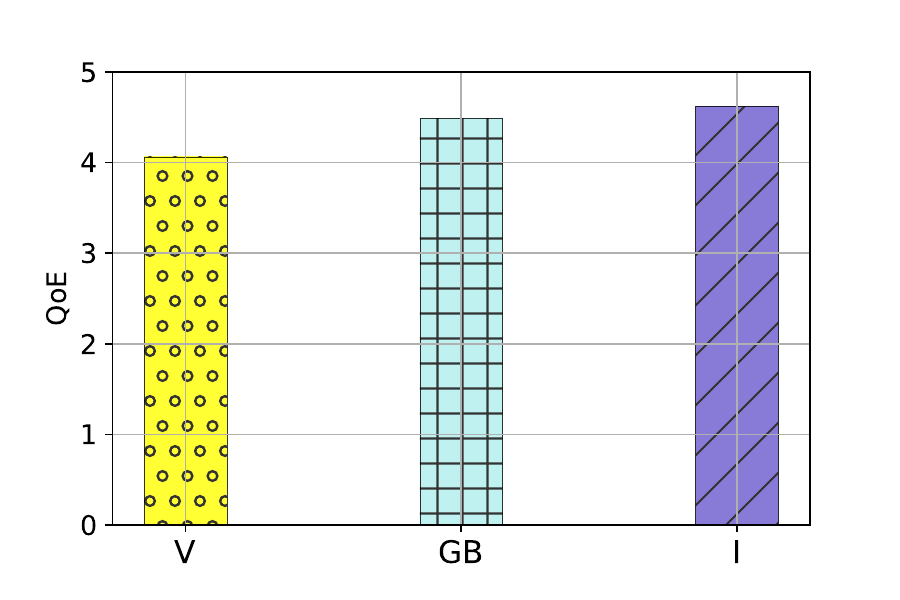}
\caption{Comparison of average QoE with varying number clients}
\label{fig:scenario2_1_qoe}
\end{minipage}\hfill
\begin{minipage}{.32\textwidth}
\centering
\includegraphics[width=1\columnwidth]{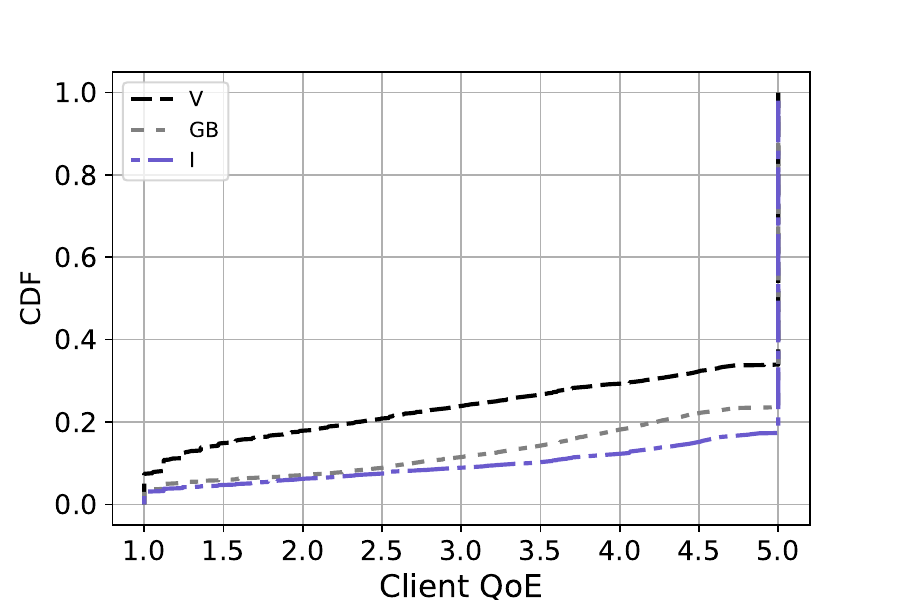}
\caption{Comparison of QoE CDF with varying number of clients}
\label{fig:scenario2_2_qoeCDF}
\end{minipage}\hfill
\begin{minipage}{.32\textwidth}
\centering
\includegraphics[width=1\columnwidth]{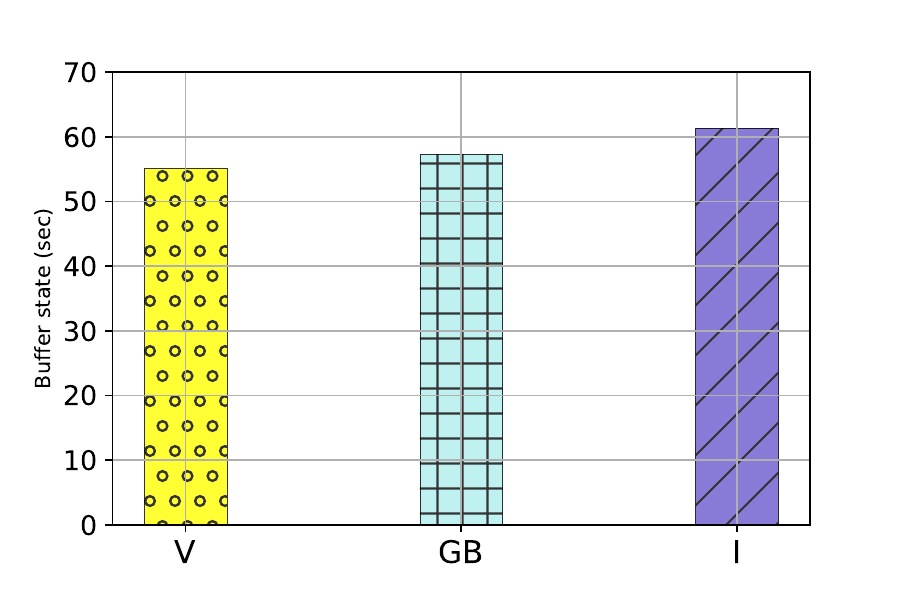}
\caption{Comparison of average buffer state with varying number of clients}
\label{fig:scenario2_3_buffer}
\end{minipage}\hfill
\vspace{-0.1in}

\begin{minipage}{.32\textwidth}
\centering
\includegraphics[width=1\columnwidth]{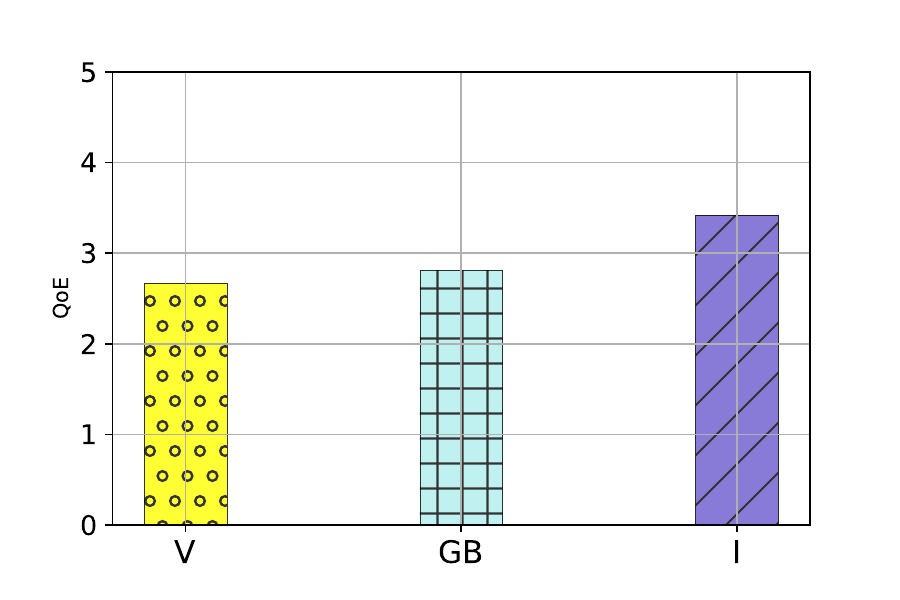}
\caption{Comparison of average QoE in a poor channel}
\label{fig:scenario3_1_qoe}
\end{minipage}\hfill
\begin{minipage}{.32\textwidth}
\centering
\includegraphics[width=1\columnwidth]{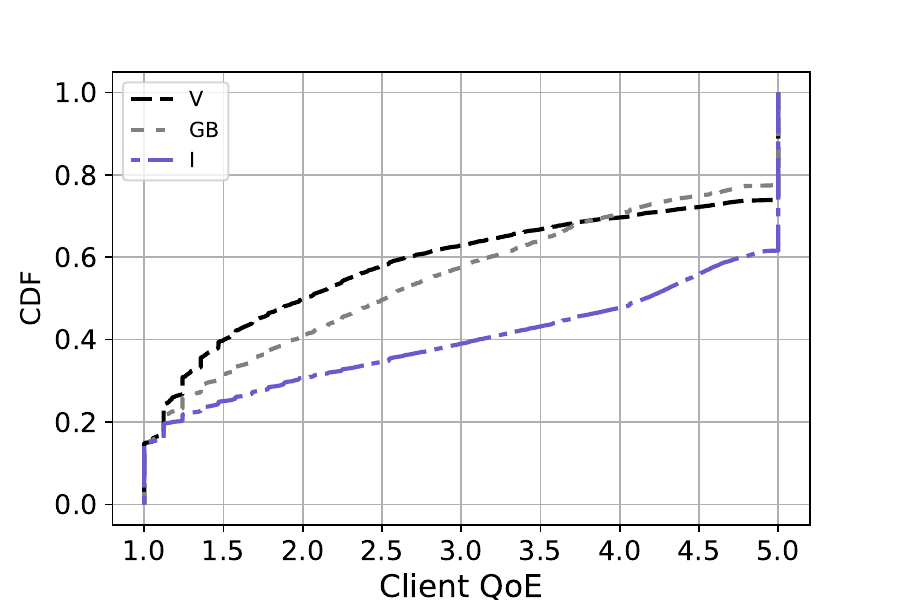}
\caption{Comparison of QoE CDF in a poor channel }
\label{fig:scenario3_2_qoeCDF}
\end{minipage}\hfill
\begin{minipage}{.32\textwidth}
\centering
\includegraphics[width=1\columnwidth]{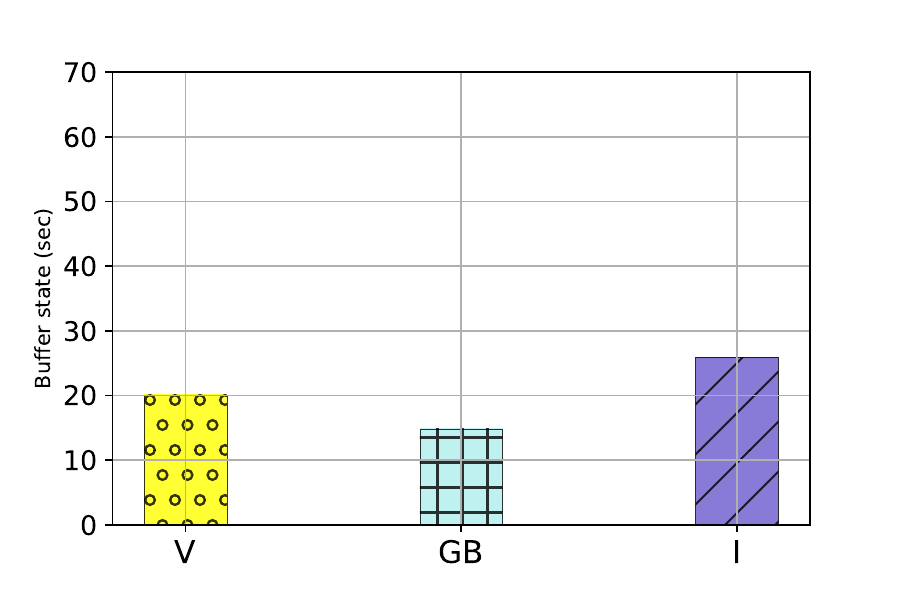}
\caption{Comparison of average buffer state in a poor channel }
\label{fig:scenario3_3_buffer}
\end{minipage}\hfill
\vspace{-0.1in}
\end{figure*}

We follow an experimental platform identical to \cite{bhattacharyya2019qflow} in which we use four Intel NUCs, with three of the NUCs hosting YouTube sessions (2 each) and the last hosting the intelligent controller and an SQL database.  We use a pub-sub approach at the database to collect state and reward data and to publish control actions.  The client NUCs have a list of popular 1080p YouTube videos, which they randomly sample, play and then flush their buffers.  We force the players to 1080p resolution across all sessions for fairness.  Relevant data such as the buffer length and number of stalls are gathered in the database and shared with the intelligent controller for decision making.  We couple this system with a WiFi access point running OpenWRT, in which we create high and low priority queues using Linux Traffic Controller (TC).   The controller uses OpenFlow experimenter messages to communicate its prioritization decisions to the access point every 10 seconds.   The  communication and computational overheads of obtaining the state information from the clients over the wireless channel and execution of simulation-trained RL policies is minimal due to the low frequency of state updates. 

We perform the evaluations below in a laboratory setting with normal background WiFi traffic.  We verified using iPerf that the throughput limits that we set on TC are actually attained, i.e., the background WiFi traffic does not significantly affect performance.  In each of the settings below, each algorithm was run for a minimum of three hours for data collection.  We verified that the collected samples do not show much statistical difference for longer collection periods.   We also performed one test in an anechoic chamber, and found no significant variation in system performance.

\vspace{0.1in}
\noindent\textit{A. Does structured RL provide high-performance?} 

%\vspace{-0.1in}
We first determine if structured RL approaches are near-optimal.  Figure \ref{fig:scenario1_1} shows the average QoE across three hours of YouTube sessions for all 7 policies. As expected, the centralized policies CH and CS do the best, with the decentralized versions DC, DC-T and I following close behind.  GB does reasonably well, but is unable to account for stalls influencing QoE, hence causing performance loss.  The vanilla approach shows why intelligent control is necessary, as it lags behind the RL approaches by full QoE unit, i.e., it is over 30\% worse.  Figure \ref{fig:scenario1_2_qoeCDF} shows the CDF of QoE samples, where the value of the RL  approach is even clearer from the fact that about 60-70\% of samples have a perfect score of QoE 5, while the other approaches are about 20\% worse.  Finally, Figure \ref{fig:scenario1_3} shows the average video buffer length, and appears to indicate that Greedy tries to equalize buffers, letting them go too low before prioritizing.  The RL-based policies all have higher average buffers, consistent with higher QoE and fewer stalls.

\vspace{0.1in}
\noindent\textit{B. Is indexing robust to a variable number of clients?} 

We next address whether the index policy I can be applied unaltered to the case where the number of clients is less than 6, which means that there is no need to train over a variable number of clients.  So we perform an experiment where we decrease number of clients from 6 to 4 over three hours.  As expected, the QoE of all three policies increases, with all three policies having an average QoE above 4 in Figure~\ref{fig:scenario2_1_qoe}. However, the overall trend is maintained as the Index policy does the best as seen in Figures \ref{fig:scenario2_1_qoe}, \ref{fig:scenario2_2_qoeCDF}, and \ref{fig:scenario2_3_buffer}.

\vspace{0.1in}
\noindent\textit{C. Does the index approach generalize to variable channels?} 

Wireless channels are subject to variability due to mobility and obstruction, and the question arises as to whether the RL approach needs to be trained separately over many channel realizations?  Since indexing simply orders clients based on their states, we wish to determine if it can be used unchanged even when the channel quality is low.  Hence, we introduce packet losses and delays using a network emulator to create channels with disturbances.  We first group clients with similar channels into clusters and divide the total available resources in a proportionally fair manner across clusters.  We then apply our service differentiation policies across members of each cluster.   We show results for a cluster seeing 10\% packet loss and 20 ms delay, and, as expected, QoE drops as seen in Figure~\ref{fig:scenario3_1_qoe}. However, the overall order still holds as Vanilla and Greedy Buffer have QoEs over 20\% worse than Index. Further, the client QoE CDF and average buffer state indicate same trend as seen in \ref{fig:scenario3_2_qoeCDF} and \ref{fig:scenario3_3_buffer}, with Indexing still performing significantly better than the others.
%https://www.overleaf.com/project/63643d0bcacaf7d90ee10820

\section{Conclusion}

In this work, we investigated the optimization of resource allocation at the wireless edge for media streaming, particularly focusing on YouTube sessions, aiming to enhance users' quality of experience (QoE). Through a constrained Markov Decision Process (CMDP) framework and a data-driven approach using constrained reinforcement learning (CRL), we showed a threshold structure in optimal policies and developed a primal-dual natural policy gradient algorithm to efficiently learn such threshold policies. Specifically, we showed that soft-threshold parametrization coupled with natural policy gradient has a fast convergence rate, and only requires single neuron training.  Simulation results demonstrate the effectiveness of the approach, with decentralized single-client decomposition learning approximately four times faster than centralized learning while maintaining similar performance, while inference is completed in about four times less time. Implementing learned policies on an intelligent controller platform and evaluating them in real-world settings, the we showed significant QoE improvement of over 30\% compared to a vanilla policy, with the proposed index policy exhibiting robustness to varying load and channel conditions while maintaining high QoE levels. 

\textbf{Limitations:}  The proposed approach relies heavily on simulation environments, and so fine tuning or adaptation based on the exact streaming application used and its relevant QoE metrics (especially if different from YouTube), as well as the mobility and occlusion of the physical environment needs to be considered.  Additionally, the scalability and deployment feasibility of the proposed intelligent controller platform in large-scale networks warrant furthers exploration.  Overall, while our work presents valuable insights on the potential of intelligent control in enhancing user experience in wireless media streaming, addressing these limitations could enhance the practical relevance.% and applicability.

% In this work we explored the value of structured reinforcement learning for solving constrained MDPs in the context of media streaming systems. The key observation was that by first relaxing the constraint, we obtain a simpler CMDP whose optimal policy both has a threshold structure,\redtext{and value function has a unimodal behavior with respect to the threshold parameter.   These two properties enable us to develop a 3-timescale policy gradient algorithm that provably converges to the global optimum, rather than the local optimum that is usually guaranteed by the policy gradient approach. }  We empirically verify that exploiting structure enables much faster and simpler approaches to learning a near-optimal policy, which can be made robust to variations in the setting via indexing.  Finally, we verify using real-world experiments our hypothesis regarding structured RL being able to robustly outperform traditional approaches in a variety of scenarios.  We are already working on translating our system to act as a RAN intelligent controller (RIC) at a cellular base station, which will form the basis of our future work.

%% file: 07-Appendix.tex
\appendix 
\section{Appendix}

\subsection{Value Iteration Algorithm}
\label{sec:appendix-Value}

Value Iteration (VI) algorithm can be used to find the optimal value function of an MDP.  Define the Bellman operator $T$ as
\begin{equation*}
    TJ(s;\lambda) = \min_{a} \sum_{s'} p(s,a,s') [c(s,s') + \lambda g(a) + \gamma J(s';\lambda)],
\end{equation*}
for any value function $J$. Then, the optimal value function $J^*$ can be found using the iteration $J_{k+1} = T J_{k}$. It is well known that $J_{k} \rightarrow J^*$, starting from any arbitrary initial value $J_{0}$. The the optimal value function $J^*$ also satisfies the Bellman equation $J = TJ$. 

\subsection{Decentralized Optimal Policy}
\label{sec:theorem-1}
\begin{proof}[Proof of Theorem \ref{eqn:decentpolthm}]
As seen in the Equation~\eqref{eqn:saddlepoint}, the Lagrangian formulation for the joint CMDP decomposes into the sum of individual costs of each client $n$. Moreover, each client's optimal solution is given by state-action probabilities corresponding to a stationary randomized policy for its individual MDP. If either $\lambda = 0$, or the constraint is satisfied with equality, for this set of randomized policies for the clients, then, complimentary slackness is satisfied. Hence, $\lambda$ is optimal for the dual problem and the joint policy obtained as in the theorem statement is optimal for the primal. Hence, $(1)$ is proved.
The proof of Theorem~\eqref{eqn:decentpolthm} can be obtained by following in similar lines to~\cite{singh2019optimal}.
\end{proof}

\subsection{Proofs of Lemma \ref{lem:nondcr1} and Lemma \ref{lem:ifflemma1}}
\label{sec:append-lemma1-2}

We first provide the proof of Lemma \ref{lem:nondcr1}.

\begin{proof}[Proof of Lemma \ref{lem:nondcr1}]
With a slight abuse of notation, we omit the dependence of value function on $\lambda$. Let $J_{k}$ be the value function from the $k$-th iterate of the VI algorithm.
Define $DJ_k(s) = J_k(s) - J_k(s - e_x)$.
Let $s := (x+1,y)$, $s' := (x+2,y)$ for any given $y$,  we will show that $DJ_k(s') \ge DJ_k(s)$, $\forall x$, and $k \geq 0$ using the principle of mathematical induction and the properties of the  value iteration (VI) algorithm. Finally, the desired result follows as a convergence of VI algorithm.

Let  $J_0(s) = 0,~\forall x, y$, as standard in the VI.  So, $DJ_{0}(s') \ge DJ_{0}(s)$ by definition.

Now, as the induction hypothesis, for a given $k$, assume that  $DJ_k(s') \ge DJ_k(s), \forall x, y$.  We now show that this inequality holds for $k+1$ i.e., $DJ_{k+1}(s') \ge DJ_{k+1}(s).$ Equivalently,
\begin{equation} \label{eqn: eqn1}
2J_{k+1}(s) \leq J_{k+1}(s')+J_{k+1}(s-e_x). 
\end{equation}
Let $a_1,a_2$ be the minimizing actions in states $s'$ and $s - e_x$ respectively. Then, with a slight abuse of notation, we define the $k$-th action value function of state $s'' = (x,y)$ for a given action $a$ by $J_{k}(s'',a) := \sum_{s'''} p(s'',a,s''') [c(s'',s''') + \lambda g(a) + \gamma J_k(s''')],~\forall x,y $. Similarly, $DJ_k(s,a) := J_k(s,a) - J_k(s - e_x,a)$.  With this, the following is immediately clear.
\begin{align*}
2J_{k+1}(s) &\leq J_{k+1}(s,a_1) + J_{k+1}(s,a_2)\\
&= J_{k+1}(s',a_1) + J_{k+1}(s-e_x,a_2)\\
&+ J_{k+1}(s,a_1) - J_{k+1}(s',a_1)\\
&+ J_{k+1}(s,a_2) - J_{k+1}(s-e_x,a_2)\\
&= J_{k+1}(s') + J_{k+1}(s-e_x)\\
&+ DJ_{k+1}(s,a_2) - DJ_{k+1}(s',a_1)
\end{align*}
Let $B \triangleq DJ_{k+1}(s,a_2) - DJ_{k+1}(s',a_1)$. In view of equation~\eqref{eqn: eqn1}, it suffices to show that $B \leq 0$.
%To this end, we consider the possible combinations of $a_1$ and $a_2$.\\ 
%We prove two cases here, the remaining follow from similar arguments and the set of assumptions on the cost model ~\eqref{assume:CostModel}.\\
%\textbf{Case $(a)$:} $a_1 = a_2 = a.$\\
By definition we have,
\begin{align*}
\hspace{0mm} 
&DJ_{k+1}(s,a_2) =J_{k+1}(s,a_2) - J_{k+1}(s-e_x,a_2)\\
&= \alpha [c(s,\mathbf{0}) - c(s-e_x,\mathbf{0}) ] \\ 
&+ (1-\alpha) \left[P_1(a_2) [c(s,s) - c(s-e_x,s-e_x)+ \gamma DJ_k(s)]\right.\\
& \left. + P_2(a_2) [c(s,s') - c(s-e_x,(s-e_x)+e_x) \right. \\ 
&\left. +  \gamma [J_k(s') - J_k((s-e_x)+e_x)]] \right.\\ 
&  \left. +P_3(a_2) [c(s,s-e_x) - c(s-e_x,s-2e_x) + \right.\\
&\left.\gamma (J_k(s-e_x)- J_k(s-2e_x))]\right]\\
&\leq  \alpha [c(s,\mathbf{0}) - c(s-e_x,\mathbf{0}) ]+\\  
& (1-\alpha) \left[P_1(a_2) [c(s,s) - c(s-e_x,s-e_x)+ \gamma DJ_k(s)] +\right.\\
& \left. P_2(a_2) [c(s,s') - c(s-e_x,(s-e_x)+e_x)+ \gamma DJ_k(s')] + \right.\\ 
&  \left.P_3(a_2) [c(s,s-e_x) - c(s-e_x,s-2e_x) + \gamma DJ_k(s-e_x)]\right]
\end{align*}
%}
where the last inequality is due to the fact that $J_k((s-e_x)+e_x) \ge J_k(s)$. This must be true, since the cost increases in $y$ for fixed buffer length, by assumption $A1$, and by assumption $A2$, that cost remains same during play period.\\
Similarly, we can write
% {\setlength{\mathindent}{-8mm}
\begin{align*}
\hspace{-3mm}
&DJ_{k+1}(s',a_1) = J_{k+1}(s',a_1) - J_{k+1}(s,a_1)\\
&\leq  \alpha [c(s',\mathbf{0}) - c(s,\mathbf{0}) ] \\ 
+(1-\alpha) &[ P_1(a_1) [c(s',s') - c(s,s) + \gamma DJ_k(s')] \\
&+  P_2(a_1) [c(s',s'+ e_x) -  c(s,s')+ 
 \gamma DJ_k(s'+ e_x)] \\ 
&+ P_3(a_1) [c(s',s) - c(s,s-e_x) + \gamma DJ_k(s)] ].
\end{align*}
%}
Now, using these two expressions, and by assumptions on the cost structure $A[1-3]$ we obtain:
% {\setlength{\mathindent}{-8mm}
% \begin{align*}
% \hspace{0mm}
% B &\leq
% \gamma(1-\alpha) P_1(a) [DJ_k(s) - DJ_k(s')] +\\
% &\hspace{5mm} \gamma (1-\alpha) P_2(a) [DJ_k(s') - DJ_k(s'+e_x)] +\\ 
% &\hspace{5mm} \gamma(1-\alpha) P_3(a) [DJ_k(s-e_x) - DJ_k(s)]\\
% &\leq 0,
% \end{align*}
% %}
% where the last inequality follows from induction step.
% \hspace{20mm}
% \textbf{Case} (b): Similarly, when $a_1 \neq a_2$, we obtain,
\begin{align*}
\hspace{0mm}
B &\leq
\gamma(1-\alpha) [P_1(a_2) DJ_k(s) - P_1(a_1) DJ_k(s') +\\
&\hspace{5mm} P_2(a_2) DJ_k(s') - P_2(a_1) DJ_k(s'+e_x) +\\ 
&\hspace{5mm} P_3(a_2) DJ_k(s-e_x) - P_3(a_1) DJ_k(s)]\\
&\leq \gamma(1-\alpha)[(DJ_k(s) - DJ_k(s')) (1-P_2(a_2) - P_3(a_1))\\
&+ P_2(a_1) (DJ_k(s') - DJ_k(s'+e_x))\\
&+ P_3(a_2) (DJ_k(s-e_x) - DJ_k(s))]\\
&\leq 0.
\end{align*}
% {\setlength{\mathindent}{-8mm}
% \begin{align*}
% \hspace{0mm}
% &B \leq \gamma (1-\alpha) \left[(P_2(2)-P_2(1)) [DJ_k(s') - DJ_k(s)] \right.\\ 
% & \left. + (P_3(2)-P_3(1)) [DJ_k(s-e_x) - DJ_k(s)] \right] \\
% & \leq 0.
% \end{align*}
%}
The result follows from induction hypothesis and the fact that $1-P_2(a_2) - P_3(a_1) = 1 - [(1-\beta)\mu(a_2) + \beta(1 - \mu(a_1))]$ is non-negative for any choice of $a_2$ and $a_1$.

Finally, as described in Appendix~\ref{sec:appendix-Value}, since the VI Algorithm converges to the optimal value function, we have that $J_k \to J^*$, as ${k \to \infty}$. Hence, we have that $J^*(s;\lambda)-J^*(s-e_x;\lambda)$, is non decreasing in $x$, for any fixed $\lambda$, and for all $x,y$.
\end{proof}

%\begin{remark}
%We can also prove the above lemma from the initial condition $J_0(x,y) = \frac{\gamma^x}{1-\gamma}c(0,0,M)$. This means that when we start with any of the above initial conditions, every iterate of VIA algorithm does have the property shown in lemma~\ref{lem:nondcr1}.
%\end{remark}

We now prove additional lemmas from which the proof of Lemma~\ref{lem:ifflemma1} follows immediately.
\begin{lem}\label{lem:ifflemma}
The optimal action in state $s := (x,y)$ in the $(k+1)^{th}$ iteration of VI is to play action $2$, if and only if, 
$(1-\beta) [J_k(s + e_x;\lambda) - J_k(s;\lambda)] + \beta [J_k(s;\lambda) - J_k(s-e_x;\lambda)] \ge r$,
where $r = c_0-\frac{\lambda}{\gamma (1-\alpha) (\mu_1 - \mu_2)},$ with $c_0 = \frac{1}{\gamma} ((1-\beta) [c(s,s)-c(s,s+e_x)] + \beta [c(s,s-e_x) - c(s,s)])$.
\end{lem}
\begin{proof}
The $(k+1)^{th}$ iterate of VI starting in state $s := (x,y)$ is given by:
% {\setlength{\mathindent}{-3mm}
\begin{align*}
&J_{k+1}(s) \\
&= \min_{a \in \{1,2\}} \{\lambda g(a) +  \sum_{\tilde{s}} p(s,a,\tilde{s}) [c(s,\tilde{s}) + \gamma J_k(\tilde{s})]\} \\
&= \min \{\lambda + \alpha c(s,\mathbf{0})+ \gamma \alpha J_k(\mathbf{0}) +  (1-\alpha)\\
& \hspace{19mm}  \left[ P_1(1) [c(s,s)+\gamma J_k(s)] + \right. \\
& \left. \hspace{20mm} P_2(1) [c(s,s+e_x) + \gamma J_k(s+e_x)] + \right.\\
&\left. \hspace{20mm} P_3(1) [c(s,s-e_x) + \gamma J_k(s-e_x)]  \right],\\
& \hspace{10mm}
\alpha c(s,\mathbf{0}) + \gamma \alpha J_k(\mathbf{0}) +  (1-\alpha)\\
& \hspace{14mm}  \left[ P_1(2) [c(s,s)+\gamma J_k(s)] + \right. \\
& \left. \hspace{15mm} P_2(2) [c(s,s+e_x) + \gamma J_k(s+e_x)] + \right.\\
&\left. \hspace{15mm} P_3(2) [c(s,s-e_x) + \gamma J_k(s-e_x)]  \right]\}.
\end{align*}
%}
Clearly, the optimal action in state $x$ at $(n+1)^{th}$ iteration is $2$, if and only if:
\begin{align*}
& \lambda \ge  (1-\alpha) [ (P_2(2) - P_2(1)) [c(s,s+e_x)  - c(s,s)] +\\
& \hspace{10mm} \gamma (P_2(2) - P_2(1)) [ J_k(s+e_x) - J_k(s)] + \\
& \hspace{9mm}  (P_3(1)-P_3(2)) [c(s,s)-c(s,s-e_x)]+ \\
& \hspace{9mm} \gamma (P_3(1)-P_3(2)) [J_k(s)-J_k(s-e_x)]].
\end{align*}
Thus, we have,
\begin{align*}
&(1-\beta) [J_k(s+e_x) - J_k(s)] +\\
& \hspace{10mm} \beta [J_k(s) - J_k(s-e_x)] \ge r,
\end{align*}
where $r = c_0-\frac{\lambda}{\gamma (1-\alpha) (\mu_1 - \mu_2)},$ with $c_0 = \frac{1}{\gamma} ((1-\beta) [c(s,s)-c(s,s+e_x)] + \beta [c(s,s-e_x) - c(s,s)])$. 
%From assumption \textbf{A2}, we observe that $c_0$ is positive.
\end{proof}
Similar to Lemma~\ref{lem:ifflemma}, we have the following lemma which gives us the necessary and sufficient condition to play action $1$. The proof is omitted, as it follows in similar lines from lemma~\ref{lem:ifflemma} .
\begin{lem}\label{lem:ifflemma2}
The optimal action in state $s := (x,y)$ in the $(k+1)^{th}$ iteration of VI is to play action $1$, if and only if, 
$(1-\beta) [J_k(s+e_x;\lambda) - J_k(s;\lambda)] + \beta [J_k(s;\lambda) - J_k(s-e_x;\lambda)] \leq r$,
where $r$ is as defined above.
\end{lem}

\begin{remark}
The above two lemmas establish that the greedy policy from each iteration of VI is of threshold type.
\end{remark}

Now, we  give the proof of Lemma \ref{lem:ifflemma1}.
\begin{proof}[Proof of Lemma \ref{lem:ifflemma1}]
The proof follows from the above two lemmas and the convergence of VI.
\end{proof}

%%%%%%%%%%%%%%%%%%%%%%%%%%%%%%%%%%%%%%%%%%%%%%%%%%%%%%%%%%%%%%%%%%%%%%%%%%%%%%%%%%%%%%%%%%%%%%%%%%%%%%%%%%%%%%%%%%%%%%%%%%%%%%%%%%
\subsection{Proof of Corollary \ref{cor:politerate}}
\label{sec:append-cor1}
\begin{proof}[Proof of Corollary \ref{cor:politerate}]
    From lemma ~\ref{lem:npgupdateform}, we have  $\theta_{n,t+1}(s) = \theta_{n,t}(s) - \frac{\eta_1}{(1 - \gamma)} [A^{\pi_{\theta_{n,t}}}_{\lambda_t}(s,a_1) - A^{\pi_{\theta_{n,t}}}_{\lambda_t}(s,a_2)]$. By definition, for a given state $s$ the probability of taking action $a_1$,
    \begin{align*}
        &\pi_{\theta_{n,t+1}}(s,a_1) = \frac{e^{\theta_{n,t+1}(s)}}{1 + e^{\theta_{n,t+1}(s)}}.\\
        &= \frac{e^{\theta_{n,t}(s)} e^{\frac{-\eta_1 (A^{\pi_{\theta_{n,t}}}_{\lambda_t}(s,a_1)-A^{\pi_{\theta_{n,t}}}_{\lambda_t}(s,a_2))}{1-\gamma}}}{1 + e^{\theta_{n,t}(s)} e^{\frac{-\eta_1 (A^{\pi_{\theta_{n,t}}}_{\lambda_t}(s,a_1)-A^{\pi_{\theta_{n,t}}}_{\lambda_t}(s,a_2))}{1-\gamma}}}.\\
        &= \frac{e^{\theta_{n,t}(s)} e^{\frac{-\eta_1 A^{\pi_{\theta_{n,t}}}_{\lambda_t}(s,a_1)}{1-\gamma}}}{e^{\frac{-\eta_1 A^{\pi_{\theta_{n,t}}}_{\lambda_t}(s,a_2)}{1-\gamma}} + e^{\theta_{n,t}(s)} e^{\frac{-\eta_1 A^{\pi_{\theta_{n,t}}}_{\lambda_t}(s,a_1)}{1-\gamma}}}\\
        &= \frac{ \pi_{\theta_{n,t}}(s,a_1) e^{\frac{-\eta_1 A^{\pi_{\theta_{n,t}}}_{\lambda_t}(s,a_1)}{1-\gamma}}}{ \pi_{\theta_{n,t}}(s,a_2) e^{\frac{-\eta_1 A^{\pi_{\theta_{n,t}}}_{\lambda_t}(s,a_2)}{1-\gamma}} + \pi_{\theta_{n,t}}(s,a_1) e^{\frac{-\eta_1 A^{\pi_{\theta_{n,t}}}_{\lambda_t}(s,a_1)}{1-\gamma}}}\\
        &= \pi_{\theta_{n,t}}(s,a_1) \frac{\exp(\frac{-\eta_1}{1-\gamma}A^{\pi_{\theta_{n,t}}}_{\lambda_t}(s,a_1))}{Z_t(s)},\\
    \end{align*}
where $Z_t(s) \triangleq \sum_{a \in \cA} \pi_{\theta_{n,t}}(s,a) \exp(\frac{-\eta_1}{1-\gamma}A^{\pi_{\theta_{n,t}}}_{\lambda_t}(s,a))$.
The fourth equality follows from normalization with $1 + e^{\theta_{n,t}(s)}$. Similarly, the policy update can be derived for action $a_2$.
\end{proof}

\subsection{Convergence of NPG}
\label{sec:convlogregobj}
In this section, we present the proof of theorem~\eqref{thm:finiteconvthm}. We first prove the necessary supporting lemmas which will be used in the later parts of the proof. To this effect, we define shorthand notation for quantities of interest useful in our analysis. For the ease of notation, we drop the client subscript $n$ and use it explicitly when needed.

%\textbf{Notation.}
We denote the discounted state visitation distribution of a policy $\pi_{\theta}$ starting at initial state $s$ as:
\begin{align}
d_{s}^{\pi_{\theta}}(s') \triangleq (1-\gamma) \sum_{t \geq 0} \gamma^t p^{\pi_{\theta}}(s(t) = s'|s(0) = s)
\end{align}
where $p^{\pi_{\theta}}(s(t) = s'| s(0) = s)$ is the probability that $s(t) = s'$ starting at an initial state $s(0) = s$, under the policy $\pi_{\theta}$.
Accordingly, we denote the discounted state visitation distribution for an initial distribution $\rho$ by
\begin{align}
d_{\rho}^{\pi_{\theta}}(s') \triangleq \E_{s \sim \rho} [(1-\gamma) \sum_{t \geq 0} \gamma^t p^{\pi_{\theta}}(s(t) = s'| s(0) = s)]
\end{align}

\begin{lem}[Improvement Lemma]\label{eqn:ImprovementLemma}
The policy iterates $\pi_{\theta_t}$ generated by algorithm~\ref{eqn:pseudoNPGalgo2} satisfy
\begin{align*}
    &J_c^{\pi_{\theta_t}}(\rho) - J_c^{\pi_{\theta_{t+1}}}(\rho)\\
    &+ \lambda_t (J_g^{\pi_{\theta_t}}(\rho) - J_g^{\pi_{\theta_{t+1}}}(\rho)) \geq \frac{1 - \gamma}{\eta_1} \E_{s \sim \rho} \log Z_t(s). \numberthis
\end{align*}
and $\E_{s \sim \rho} \log Z_t(s) \geq 0$ for any initial state distribution $\rho$.
\end{lem}
\begin{proof}
From the performance difference lemma~\cite[Lemma $2$]{agarwal2021theory} we have,
\begin{align*}
&J_c^{\pi_{\theta_{t+1}}}(\rho) - J_c^{\pi_{\theta_t}}(\rho)
 = \frac{\E_{s \sim d_{\rho}^{\pi_{\theta_{t+1}}}, a \sim  \pi_{\theta_{t+1}}(s, \cdot)} [A_{c}^{\pi_{\theta_t}}(s,a)]}{1-\gamma}\\
&= \frac{1}{1-\gamma} \E_{s \sim d_{\rho}^{\pi_{\theta_{t+1}}}} \sum_{a \in \cA}  \pi_{\theta_{t+1}}(s, a) [A_{c}^{\pi_{\theta_t}}(s,a)]\\
%= \frac{1}{1-\gamma} \E_{s \sim d_{\rho}^{\pi_{\theta_{t+1}}}} \sum_{a \in \cA}  \pi_{\theta_{t+1}}(s, a) &[A_{c}^{\pi_{\theta_t}}(s,a)]\\
&= \frac{1}{\eta_1} \E_{s \sim d_{\rho}^{\pi_{\theta_{t+1}}}} \sum_{a \in \cA}  \pi_{\theta_{t+1}}(s, a) \log(\frac{\pi_{\theta_t}(s,a)}{\pi_{\theta_{t+1}}(s,a)Z_t(s)}) \nonumber\\
&- \frac{\lambda_t}{1-\gamma} \E_{s \sim d_{\rho}^{\pi_{\theta_{t+1}}}} \sum_{a \in \cA}  \pi_{\theta_{t+1}}(s, a) A_{g}^{\pi_{\theta_t}}(s,a)\\
&= \frac{1}{\eta_1} \E_{s \sim d_{\rho}^{\pi_{\theta_{t+1}}}} \sum_{a \in \cA}  \pi_{\theta_{t+1}}(s, a) \log(\frac{\pi_{\theta_t}(s,a)}{\pi_{\theta_{t+1}}(s,a)Z_t(s)}) \nonumber\\
&\quad- \lambda_t ( J_g^{\pi_{\theta_{t+1}}}(\rho) - J_g^{\pi_{\theta_t}}(\rho)).
\end{align*}
The last equality comes from performance difference lemma~\cite[Lemma $2$]{agarwal2021theory} with respect to the discounted utility function. Rewriting the expression we have,
\begin{align*}
&J_c^{\pi_{\theta_t}}(\rho) - J_c^{\pi_{\theta_{t+1}}}(\rho) + \lambda_t (J_g^{\pi_{\theta_t}}(\rho) - J_g^{\pi_{\theta_{t+1}}}(\rho)) \\
 &= \frac{1}{\eta_1} \E_{s \sim d_{\rho}^{\pi_{\theta_{t+1}}}} \sum_{a \in \cA}  \pi_{\theta_{t+1}}(s, a) \log(\frac{\pi_{\theta_{t+1}}(s,a)Z_t(s)}{\pi_{\theta_t}(s,a)}) \\
 &= \frac{1}{\eta_1} \E_{s \sim d_{\rho}^{\pi_{\theta_{t+1}}}}  D_{KL} (\pi_{\theta_{t+1}} || \pi_{\theta_t}) + \frac{1}{\eta_1} \E_{s \sim d_{\rho}^{\pi_{\theta_{t+1}}}} \log Z_t(s) \\
 &\geq \frac{1}{\eta_1} \E_{s \sim d_{\rho}^{\pi_{\theta_{t+1}}}} \log Z_t(s)
\end{align*}
The last inequality follows from the fact that KL-divergence is non-negative. Therefore, it suffices to show that $\log Z_t(s) \geq 0$. From the definition of $Z_t(s)$ and the concavity of the log function we have,
\begin{align*}
    \log Z_t(s) &= \log (\sum_{a \in \cA} \pi_{\theta_t}(s,a) \exp(\frac{-\eta_1}{1-\gamma}A_{\lambda_t}(s,a)))\\
                &\geq \sum_{a \in \cA} \pi_{\theta_t}(s,a) \log \exp(\frac{-\eta_1}{1-\gamma}A_{\lambda_t}(s,a))\\
                &\geq \frac{-\eta_1}{1-\gamma} \sum_{a \in \cA} \pi_{\theta_t}(s,a) (A_r^{\pi_{\theta_t}}(s,a) + \lambda_t A_g^{\pi_{\theta_t}}(s,a))\\
                &= 0.
\end{align*}
The last equality uses the fact that the expected advantage in any state $s$ is $0$. Finally, using the fact that $d_{\rho}^{\pi_{\theta_{t+1}}}(s) \geq \rho(s) (1 - \gamma)~ \forall s \in \cS$, we have the result.
\end{proof}

Let $\pi^*$ be the optimal policy.
\begin{lem}[Bounded Average performance]\label{lem:boundedavgperf}
     The policy iterates $\pi_{\theta_t}$, and the dual variable iterates $\lambda_t$ of algorithm~\ref{eqn:pseudoNPGalgo2} satisfy
    \begin{align*}
        \frac{1}{T} \sum_{t = 0}^{T-1} J_c^{\pi_{\theta_t}}(\rho) - J_c^{\pi^*}(\rho) + \frac{1}{T} \sum_{t = 0}^{T-1} \lambda_t (J_g^{\pi_{\theta_t}}(\rho) - J_g^{\pi^*}(\rho))\\
        \leq \frac{\log |\cA|}{\eta_1 T} + \frac{1}{(1-\gamma)^2T} + \frac{N\eta_2}{(1-\gamma)^3}. \numberthis
    \end{align*}
\end{lem}
\begin{proof}
Using the performance difference lemma now with respect to the optimal policy $\pi^*$,
\begin{align*}
 &J_c^{\pi^*}(\rho) -  J_c^{\pi_{\theta_t}}(\rho) = \frac{1}{1-\gamma} \E_{s \sim d_{\rho}^{\pi^*}} \sum_{a \in \cA}  \pi^*(s, a) A_{r}^{\pi_{\theta_t}}(s,a)\\
 &= \frac{1}{\eta_1} \E_{s \sim d_{\rho}^{\pi^*}} \sum_{a \in \cA}  \pi^* (s, a) \log(\frac{\pi_{\theta_t}(s,a)}{\pi_{\theta_{t+1}}(s,a)Z_t(s)})\\
 &- \frac{\lambda_t}{1-\gamma} \E_{s \sim d_{\rho}^{\pi^*}} \sum_{a \in \cA}  \pi^*(s, a) A_{g}^{\pi_{\theta_t}}(s,a)\\
 &= \frac{1}{\eta_1} \E_{s \sim d_{\rho}^{\pi^*}} [ D_{KL} (\pi^* || \pi_{\theta_{t+1}}) - D_{KL} (\pi^* || \pi_{\theta_t})]\\
 &- \frac{1}{\eta_1} \E_{s \sim d_{\rho}^{\pi^*}} \log Z_t(s) - \lambda_t (J_g^{\pi^*}(\rho) -  J_g^{\pi_{\theta_t}}(\rho)).
 \end{align*}
 Rearranging the terms of above expression and taking the summation over $T$ iterations we have,
%  \begin{align}
% &J_c^{\pi_{\theta_t}}(\rho) - J_c^{\pi^*}(\rho)\\ 
% &= \frac{1}{\eta_1} \E_{s \sim d_{\rho}^{\pi^*}} [ D_{KL} (\pi^* || \pi_{\theta_{t}}) - D_{KL} (\pi^* || \pi_{\theta_{t+1}})]\\
% &+ \frac{1}{\eta_1} \E_{s \sim d_{\rho}^{\pi^*}} \log Z_t(s) - \lambda_t (J_g^{\pi_{\theta_t}}(\rho) - J_g^{\pi^*}(\rho)).
% \end{align}
% Taking the sum over $T$ iterations we have
\begin{align*}
    &\frac{1}{T} \sum_{t = 0}^{T-1} J_c^{\pi_{\theta_t}}(\rho) - J_c^{\pi^*}(\rho)\\
    &= \frac{1}{\eta_1T} \sum_{t = 0}^{T-1} \E_{s \sim d_{\rho}^{\pi^*}} [ D_{KL} (\pi^* || \pi_{\theta_{t}}) - D_{KL} (\pi^* || \pi_{\theta_{t+1}})]\\
    &+ \frac{1}{\eta_1T} \sum_{t = 0}^{T-1}\E_{s \sim d_{\rho}^{\pi^*}} \log Z_t(s) - \frac{1}{T} \sum_{t = 0}^{T-1} \lambda_t (J_g^{\pi_{\theta_t}}(\rho) - J_g^{\pi^*}(\rho))
\end{align*}
With a slight abuse of notation, using the optimal state visitation distribution denoted by $\rho^* \triangleq d_{\rho}^{\pi^*}$, from the improvement lemma~\ref{eqn:ImprovementLemma} it follows that
\begin{align}
    &J_c^{\pi_{\theta_t}}(\rho^*) - J_c^{\pi_{\theta_{t+1}}}(\rho^*) + \lambda_t (J_g^{\pi_{\theta_t}}(\rho^*) - J_g^{\pi_{\theta_{t+1}}}(\rho^*)) \nonumber\\
    &\geq \frac{1 - \gamma}{\eta_1} \E_{s \sim \rho^*} \log Z_t(s).
\end{align}

Using the above inequality, we upper bound the average performance difference as follows
\begin{align*}
    &\frac{1}{T} \sum_{t = 0}^{T-1} J_c^{\pi_{\theta_t}}(\rho) - J_c^{\pi^*}(\rho)\\
    &\leq \frac{1}{\eta_1T} \sum_{t = 0}^{T-1} \E_{s \sim \rho^*} [ D_{KL} (\pi^* || \pi_{\theta_{t}}) - D_{KL} (\pi^* || \pi_{\theta_{t+1}})]\\ 
    &+ \frac{1}{(1-\gamma)T} \sum_{t = 0}^{T-1} J_c^{\pi_{\theta_t}}(\rho^*) - J_c^{\pi_{\theta_{t+1}}}(\rho^*)\\
    &+ \frac{\lambda_t }{(1-\gamma)T} \sum_{t = 0}^{T-1}(J_g^{\pi_{\theta_t}}(\rho^*) - J_g^{\pi_{\theta_{t+1}}}(\rho^*))\\
    &- \frac{1}{T} \sum_{t = 0}^{T-1} \lambda_t (J_g^{\pi_{\theta_t}}(\rho) - J_g^{\pi^*}(\rho))\\
    &\leq \frac{1}{\eta_1T}  \E_{s \sim \rho^*} D_{KL} (\pi^* || \pi_{\theta_{0}}) + \frac{1}{(1-\gamma)T} J_c^{\pi_{\theta_0}}(\rho^*)\\ 
    &+ \frac{1}{(1-\gamma)T} \sum_{t = 0}^{T-1} \lambda_t (J_g^{\pi_{\theta_t}}(\rho^*) - J_g^{\pi_{\theta_{t+1}}}(\rho^*))\\
    &- \frac{1}{T} \sum_{t = 0}^{T-1} \lambda_t (J_g^{\pi_{\theta_t}}(\rho) - J_g^{\pi^*}(\rho))\\
    &\leq \frac{1}{\eta_1T}  \E_{s \sim \rho^*} D_{KL} (\pi^* || \pi_{\theta_{0}}) + \frac{1}{(1-\gamma)T} J_c^{\pi_{\theta_0}}(\rho^*)\\
    &+ \frac{1}{(1-\gamma)T} \sum_{t = 1}^{T-1} (\lambda_t - \lambda_{t-1})J_g^{\pi_{\theta_t}}(\rho^*) + \frac{\lambda_0}{(1-\gamma)T}J_g^{\pi_{\theta_0}}(\rho^*)\\
    &- \frac{1}{T} \sum_{t = 0}^{T-1} \lambda_t (J_g^{\pi_{\theta_t}}(\rho) - J_g^{\pi^*}(\rho))\\
    &\leq  \frac{\log |\cA|}{\eta_1 T} + \frac{1}{(1-\gamma)^2T} + \frac{\eta_2N}{(1-\gamma)^3} \nonumber\\
    &- \frac{1}{T} \sum_{t = 0}^{T-1} \lambda_t (J_g^{\pi_{\theta_t}}(\rho) - J_g^{\pi^*}(\rho)).
\end{align*}
The second inequality follows from evaluating the telescoping terms and dropping the negative terms. Finally, in the last inequality, we use the fact that the KL divergence $D_{KL}(p || q)$ where $p,q \in \triangle_{\cA}$ and $q$ is the uniform distribution is upper bounded by $\log \abs{\cA};$ the value function is upper bounded by $\frac{1}{1-\gamma}$ since costs are bounded; from the dual update in algo~\eqref{eqn:pseudoNPGalgo2} we have $\abs{\lambda_t - \lambda_{t-1}} \leq \frac{\eta_2N}{1-\gamma}$. 
\end{proof}

We are now ready to prove theorem~\eqref{thm:finiteconvthm}. We first prove the optimality gap, and then prove a series of supporting lemmas from which the constraint violation gap is derived thereafter. Note that the above technical lemmas are true for each client $n$. Using the shorthand notation for the cumulative objective functions $\bar{J}_c^{\pi} \triangleq \sum_{n=1}^N J_c^{\pi_n}(\rho_n), \bar{J}_g^{\pi} \triangleq \sum_{n=1}^N J_g^{\pi_n}(\rho_n)$, $\bar{J}_{\lambda}^{\pi} \triangleq \bar{J}_c^{\pi} + \lambda (\bar{J}_g^{\pi} - \bar{K})$.\\
\begin{proof}[Bounding the optimality gap in theorem~\eqref{thm:finiteconvthm}]

From Lemma~\eqref{lem:boundedavgperf}, for each client $n$ we have,
    \begin{align*}
    \frac{1}{T} \sum_{t = 0}^{T-1} J_c^{\pi_{\theta_{n,t}}}(\rho) - J_c^{\pi_n^*}(\rho) &\leq  \frac{\log |\cA|}{\eta_1 T} + \frac{1}{(1-\gamma)^2T}\\ 
    + \frac{\eta_2N}{(1-\gamma)^3}
    &- \frac{1}{T} \sum_{t = 0}^{T-1} \lambda_t (J_g^{\pi_{\theta_{n,t}}}(\rho) - J_g^{\pi_n^*}(\rho))
    \end{align*}
Summing over all clients,
    \begin{align}\label{eqn:bdavgperfineql}
    &\frac{1}{T} \sum_{t = 0}^{T-1} \sum_{n = 1}^N J_c^{\pi_{\theta_{n,t}}}(\rho_n) - J_c^{\pi_n^*}(\rho_n)
    \leq  \frac{N\log |\cA|}{\eta_1 T} + \frac{N}{(1-\gamma)^2T} \nonumber\\ 
    &+ \frac{\eta_2N^2}{(1-\gamma)^3}
    - \frac{1}{T} \sum_{t = 0}^{T-1} \sum_{n = 1}^N \lambda_t (J_g^{\pi_{\theta_{n,t}}}(\rho_n) - J_g^{\pi_n^*}(\rho_n))
    \end{align}
To derive the optimality gap, it suffices to upper bound the last term in above inequality. From the dual update in the algorithm~\ref{eqn:pseudoNPGalgo2},
\begin{align*}
&0 \leq \lambda_T^2 = \sum_{t = 0}^{T-1} (\lambda_{t+1}^2 - \lambda_t^2)\\
&= \sum_{t = 0}^{T-1} (\cP_{\Lambda}(\lambda_t + \eta_2 (\sum_{n = 1}^N J_g^{\pi_{\theta_{n,t}}}(\rho_n) - \bar{K})))^2 - \lambda_t^2\\
&\leq \sum_{t = 0}^{T-1} (\lambda_t + \eta_2(\sum_{n = 1}^NJ_g^{\pi_{\theta_{n,t}}}(\rho_n) - \bar{K}))^2 - \lambda_t^2\\
&\leq 2\eta_2 \sum_{t = 0}^{T-1} \lambda_t(\sum_{n = 1}^NJ_g^{\pi_{\theta_{n,t}}}(\rho_n) - \bar{K})\\
&+ \eta_2^2 \sum_{t = 0}^{T-1}(\sum_{n = 1}^NJ_g^{\pi_{\theta_{n,t}}}(\rho_n) - \bar{K})^2\\
&\leq 2\eta_2 \sum_{t = 0}^{T-1} \sum_{n = 1}^N \lambda_t(J_g^{\pi_{\theta_{n,t}}}(\rho_n) -  J_g^{\pi_n^*}(\rho_n)) +  \frac{\eta_2^2TN^2}{(1-\gamma)^2}.
\end{align*}
Therefore we have,
\begin{equation}
 - \frac{1}{T} \sum_{t = 0}^{T-1} \sum_{n = 1}^N \lambda_t (J_g^{\pi_{\theta_{n,t}}}(\rho_n) - J_g^{\pi_n^*}(\rho_n)) \leq  \frac{\eta_2N^2}{2(1-\gamma)^2}.
\end{equation}
The second inequality follows from the dual update in ~\eqref{eqn:pseudoNPGalgo2}, and using the upper bound for the value function as before. Finally, for $\eta_1 = \log |\cA|$ and $\eta_2 = \frac{1-\gamma}{N\sqrt{T}}$, we have
\begin{align*}
&\frac{1}{T} \sum_{t = 0}^{T-1} \sum_{n = 1}^N J_c^{\pi_{\theta_{n,t}}}(\rho_n) - J_c^{\pi_n^*}(\rho_n)\\
&\leq  \frac{N\log |\cA|}{\eta_1 T} + \frac{N}{(1-\gamma)^2T} + \frac{\eta_2N^2}{(1-\gamma)^3} + \frac{\eta_2N^2}{2(1-\gamma)^2}.\\
&\leq  \frac{N\log |\cA|}{\eta_1 T} + \frac{N}{(1-\gamma)^2T} + \frac{2\eta_2N^2}{(1-\gamma)^3}.\\
&\leq  \frac{N}{T} + \frac{N}{(1-\gamma)^2T} + \frac{2N}{(1-\gamma)^2 \sqrt{T}} \leq \frac{4N}{(1-\gamma)^2 \sqrt{T}}.
\end{align*}
\end{proof}

\begin{lem}\label{lem:boundoptimaldual}
Given that assumption~\eqref{eqn:slatercond} holds true. Then $0 \leq \lambda^* \leq \frac{(\bar{J}_c^{\bar{\pi}} - \bar{J}_c^{\pi^*})}{\xi}$.
\end{lem}
The proof is similar to that of \cite[Lemma $1$]{ding2020npg}.
% \begin{proof} We begin by showing that the super-level sets of the dual function are bounded i.e. for any $\{\lambda \in \R_+ | D(\lambda) \geq \alpha\}$, it holds that $\lambda \leq \frac{1}{\xi}(\bar{J}_c^{\bar{\pi}} - \alpha)$. If $\lambda \in \{\lambda \in \R_+ | D(\lambda) \geq \alpha\}$, then $\alpha \leq \bar{J}_c^{\bar{\pi}} + \lambda (\bar{J}_g^{\bar{\pi}} - \bar{K}) \leq \bar{J}_c^{\bar{\pi}} - \lambda \xi$.

% Setting $\alpha = \bar{J}_c^{\pi^*}$ and since $D(\lambda^*) = \bar{J}_c^{\pi^*}$ from strong duality, the result follows.
% \end{proof}

\begin{lem}\label{lem:weakconstrainbound}
Given assumption~\eqref{eqn:slatercond} holds true. For any $M \geq 2 \lambda^*$, if there exists a policy $\pi = \otimes \pi_n$, and $\delta > 0$ such that $\bar{J}_c^{\pi} - \bar{J}_c^{\pi^*} + M (\bar{J}_g^{\pi} - \bar{K})^+ \leq \delta$, then $(\bar{J}_g^{\pi} - \bar{K})^+ \leq \frac{2\delta}{M}$.
\end{lem}
The proof is similar to that of \cite[Lemma $2$]{ding2020npg}.

\textsc{Bounding the constraint violation in theorem~\eqref{thm:finiteconvthm}:}
The joint policy iterate at time $t$ is given by $\pi_t \triangleq \otimes \pi_{\theta_{n,t}}$. From the dual update in ~\eqref{eqn:pseudoNPGalgo2}, for any $\lambda \in \Lambda$.
\begin{align*}
    \abs{\lambda_{t+1} - \lambda}^2 &\leq \abs{\lambda_t + \eta_2(\bar{J}_g^{\pi_t} - \bar{K}) - \lambda}^2.\\
    &\leq  \abs{\lambda_{t} - \lambda}^2 +  2\eta_2(\bar{J}_g^{\pi_t} - \bar{K})(\lambda_{t} - \lambda) + \frac{\eta_2^2N^2}{(1-\gamma)^2}
\end{align*}
Averaging over iterations $t=0$ to $t=T-1$ yields,
\begin{align*}
0 &\leq \frac{1}{T} \abs{\lambda_{T} - \lambda}^2\\
&\leq \frac{1}{T} \abs{\lambda_{0} - \lambda}^2  + \frac{2\eta_2}{T} \sum_{t=0}^{T-1}(\bar{J}_g^{\pi_t} - \bar{K})(\lambda_{t} - \lambda) + \frac{\eta_2^2N^2}{(1-\gamma)^2}.\\
\end{align*}
Furthermore we have,
\begin{equation}\label{eqn:ineq1}
     \frac{1}{T} \sum_{t=0}^{T-1}(\bar{K} - \bar{J}_g^{\pi_t})  (\lambda_{t} - \lambda) \leq \frac{1}{2\eta_2T} \abs{\lambda_{0} - \lambda}^2 + \frac{\eta_2N^2}{2(1-\gamma)^2}.\\
\end{equation}
Since $\bar{J}_g^{\pi^*} \leq \bar{K}$, from inequality~\eqref{eqn:bdavgperfineql} it follows that
\begin{align*}
\frac{1}{T} \sum_{t=0}^{T-1} \bar{J}_c^{\pi_t} -  \bar{J}_c^{\pi^*} 
&\leq \frac{N\log |\cA|}{\eta_1 T} + \frac{N}{(1-\gamma)^2T} + \frac{\eta_2N^2}{(1-\gamma)^3}  \nonumber\\ 
&-\frac{1}{T} \sum_{t=0}^{T-1} \lambda_t (\bar{J}_g^{\pi_t} - \bar{K})\\
\frac{1}{T} \sum_{t=0}^{T-1} \bar{J}_c^{\pi_t} -  \bar{J}_c^{\pi^*} 
&+\frac{1}{T} \sum_{t=0}^{T-1} \lambda (\bar{J}_g^{\pi_t} - \bar{K})
\leq \frac{N\log |\cA|}{\eta_1 T}\nonumber\\
+ \frac{N}{(1-\gamma)^2T} &+ \frac{\eta_2N^2}{(1-\gamma)^3}  -\frac{1}{T} \sum_{t=0}^{T-1} (\lambda_t - \lambda) (\bar{J}_g^{\pi_t} - \bar{K})
\end{align*}

Combined with inequality~\eqref{eqn:ineq1}, we have
\begin{align}
&\frac{1}{T} \sum_{t=0}^{T-1} \bar{J}_c^{\pi_t} -  \bar{J}_c^{\pi^*} 
+\frac{1}{T} \sum_{t=0}^{T-1} \lambda (\bar{J}_g^{\pi_t} - \bar{K}) \nonumber\\
&\leq \frac{N\log |\cA|}{\eta_1 T} + \frac{N}{(1-\gamma)^2T} + \frac{\eta_2N^2}{(1-\gamma)^3} + \frac{1}{2\eta_2T} \abs{\lambda_{0} - \lambda}^2 \nonumber\\
&+ \frac{\eta_2N^2}{2(1-\gamma)^2}.
\end{align}

Note that there exists a policy $\pi''$ such that $\bar{J}_c^{\pi''} = \frac{1}{T} \sum_{t=0}^{T-1} \bar{J}_c^{\pi_t}$, $\bar{J}_g^{\pi''} = \frac{1}{T} \sum_{t=0}^{T-1} \bar{J}_g^{\pi_t}$. Using the fact that $\lambda^* \leq \frac{N}{(1-\gamma)\xi}$ from lemma~\eqref{lem:boundoptimaldual}, setting $\lambda = \frac{2N}{(1-\gamma)\xi} \mathbf{1}_{(\bar{J}_g^{\pi''} > \bar{K})}$, it follows from lemma~\eqref{lem:weakconstrainbound} that,

\begin{align*}
 &(\frac{1}{T} \sum_{t=0}^{T-1} \bar{J}_g^{\pi_t} - \bar{K})^+ = (\bar{J}_g^{\pi''} - \bar{K})^+ \\
&\leq \frac{\xi\log |\cA|}{\eta_1 T} + \frac{\xi}{(1-\gamma)T} + \frac{\eta_2\xi N}{(1-\gamma)^2} + \frac{2N}{\eta_2(1-\gamma)\xi T} \nonumber\\
&+ \frac{\eta_2\xi N}{2(1-\gamma)}. 
\end{align*}
\begin{multicols*}{2}
\raggedcolumns
\text{Substituting in the values for the step-sizes $\eta_1 = \log |\cA|$, and}\\
\text{ $\eta_2 = \frac{1-\gamma}{N\sqrt{T}}$, we have}
\begin{align*}
&(\frac{1}{T} \sum_{t=0}^{T-1} \bar{J}_g^{\pi_t} - \bar{K})^+ \leq \frac{\xi}{T} + \frac{\xi}{(1-\gamma)T} + \frac{\xi}{(1-\gamma) \sqrt{T}}\\
&+ \frac{2N^2}{(1-\gamma)^2\xi \sqrt{T}} + \frac{\xi}{2\sqrt{T}} \leq \frac{(2/\xi + 4\xi)N^2}{(1-\gamma)^2\sqrt{T}}.
\end{align*}
\hspace{80mm}
\qedsymbol
\end{multicols*}